\documentclass[11pt]{article}
\usepackage{fullpage}
\usepackage{subfig}
\usepackage[usenames,dvipsnames]{xcolor}
\usepackage[linktocpage=true,
	pagebackref=true,colorlinks,
	linkcolor=BrickRed,citecolor=blue,
	bookmarks,bookmarksopen,bookmarksnumbered]
	{hyperref}
	
	\usepackage{latexsym,graphicx,epsfig,color}
\usepackage{amsfonts,amssymb,amsmath,amsthm,amstext}
\usepackage{enumitem}
\setitemize{itemsep=2pt,topsep=0pt,parsep=0pt}
\usepackage{url,setspace}
\usepackage{multirow}
\usepackage[most]{tcolorbox}
\usepackage{lipsum}
\usepackage{rotating}
\usepackage{makeidx}
\usepackage{tikz}
\usepackage{accents}
\usepackage{xspace}
\usepackage{algorithm,algpseudocode}
\usepackage{bm}
\usepackage{changepage} 	
\usepackage{thmtools,thm-restate}
\usepackage{sectsty}

\newcommand{\IGNORE}[1]{}
\allowdisplaybreaks

\makeindex

\sectionfont{\large} \subsectionfont{\normalsize}

\newtheorem{theorem}{Theorem}[section]
\newtheorem{claim}[theorem]{Claim}

\newtheorem{lemma}[theorem]{Lemma}
\newtheorem{corollary}[theorem]{Corollary}

\theoremstyle{definition}

\usepackage{thmtools,thm-restate} 

\newcommand{\halmos}{}

\newif\ifFULL

\FULLtrue

\newcounter{sahilnote}
\newcounter{dannynote}

\newcommand{\ProbeMax}{\textsf{ProbeMax}\xspace}
\newcommand{\ProbeTopr}{\textsf{ProbeTop}-$r$\xspace}
\newcommand{\Prophets}{\textsf{Free-Order Prophets}\xspace}
\newcommand{\ProphetsIP}{\textsf{Prophets}\xspace}
\newcommand{\Pandora}{\textsf{Pandora's Box with Commitment}\xspace}
\newcommand{\Santa}{\textsf{Multi-Dimensional Santa Claus}\xspace}
\newcommand{\Santano}{\textsf{Dimensional Santa Claus}\xspace}
\newcommand{\SingleSantano}{\textsf{One-Dimensional Santa Claus}\xspace}

\newcommand{\calE}{{\cal E}}
\newcommand{\OPT}{\textsc{OPT}}
\newcommand{\calF}{\mathcal{F}}
\newcommand{\calM}{\mathcal{M}}
\newcommand{\calC}{\mathcal{C}}
\newcommand{\calT}{\mathcal{T}}
\newcommand{\tcalT}{\tilde{\calT}}

\newcommand{\calV}{\mathcal{V}}
\newcommand{\eps}{\epsilon}
\newcommand{\reals}{\mathbb{R}}

\newcommand{\one}{{1}\xspace}
\newcommand{\poly}{\operatorname{poly}}
\newcommand{\baseV}{\textsf{BaseVal}\xspace}
\newcommand{\baseG}{\textsf{BaseGuess}\xspace}
\newcommand{\deltaGuess}{\textsf{DeltaGuess}\xspace}

\newcommand{\pr}[1]{{\mathbb P} \left( #1 \right)}
\newcommand{\prpar}[1]{{\mathbb P} ( #1 )}
\newcommand{\ex}[1]{{\mathbb E} \left[ #1 \right]}
\newcommand{\expar}[1]{{\mathbb E} [ #1 ]}
\newcommand{\var}{j}
\newcommand{\block}{i}
\newcommand{\jump}{K}

\begin{document}

\begin{titlepage}

\title{Efficient Approximation Schemes for \\
Stochastic Probing and Selection-Stopping Problems}

\author{%
Danny Segev\thanks{Department of Statistics and Operations Research, School of Mathematical Sciences, Tel Aviv University, Tel Aviv 69978, Israel. Email: segevdanny@tauex.tau.ac.il.  Supported by Israel Science Foundation grant 1407/20.}%
\and Sahil Singla\thanks{School of Computer Science, Georgia Institute of Technology,  Atlanta, Georgia, USA. Email: ssingla@gatech.edu. Supported in part by NSF awards CCF-2327010 and CCF-2440113.}
}

\date{}

\maketitle

\begin{abstract} 
In this paper, we propose a general framework to design \emph{efficient} polynomial time approximation schemes (EPTAS) for fundamental  stochastic combinatorial optimization problems. Given an error parameter $\epsilon>0$, such algorithmic schemes attain a $(1-\epsilon)$-approximation in  $t(\epsilon)\cdot \poly(|{\cal I}|)$ time, where $t(\cdot)$ is a function that depends only on $\epsilon$ and $|{\cal I}|$ denotes the input length. Technically speaking, our approach relies on presenting tailor-made reductions to a newly-introduced multi-dimensional Santa Claus problem. Even though the single-dimensional version of this problem is already known to be APX-Hard, we prove that an EPTAS can be designed  for a constant number of machines and dimensions, which hold for each of our applications.

To demonstrate the versatility of our framework, we first study selection-stopping settings to derive an EPTAS for the Free-Order Prophets problem [Agrawal et al., EC~'20] and for its cost-driven generalization, Pandora's Box with Commitment [Fu et al., ICALP~'18]. These results constitute the first  approximation schemes in the non-adaptive setting and improve on known \emph{inefficient} polynomial time approximation schemes (PTAS) for their adaptive variants. Next, turning our attention to stochastic probing problems,  we obtain an EPTAS for the adaptive ProbeMax problem as well as for its non-adaptive counterpart; in both cases, state-of-the-art approximability results have been  inefficient PTASes  [Chen et al., NIPS~'16; Fu et al., ICALP~'18].
\end{abstract}

\ifFULL

\newpage

{\small
\begin{spacing}{0.01}
   \tableofcontents
\end{spacing}
}
\else
\fi

\end{titlepage}

\newpage

\ifFULL
\setcounter{page}{3}
\fi

\section{Introduction.}\label{sec:intro}

The field of combinatorial optimization traditionally deals with computational problems where we are given an objective function $f: 2^{[n]}\rightarrow \reals$ on $n$ elements along with certain feasibility constraints $\calF \subseteq 2^{[n]}$; our goal is to identify in polynomial time a set $S \in \calF$ that maximizes $f(S)$, potentially in an approximate way. In the last two decades, there has been a great deal of interest in studying combinatorial optimization problems under various notions of uncertainty.  In particular, a frequent meta-question in this context is: Can we handle objective functions that involve random variables, when our algorithm only has access to their probability distributions? 

For concreteness, consider an interviewing scenario where a firm wishes to hire one of $n$ candidates. This setting corresponds to the simplest non-trivial feasibility set   $\calF=\{ \{i\} \mid i\in [n] \}$. Clearly, without any form of randomness, the problem is trivial, as we can hire  the highest value candidate. Now, suppose the value of each candidate $i$ is represented by a non-negative random variable $X_i$ independently drawn  from some known \emph{element-dependent} distribution. If we can only afford conducting $k<n$ interviews,  which candidates should be interviewed? Formally, in the \emph{\ProbeMax} problem we  \emph{probe} a set $S \subseteq [n]$ of size at most $k$, with the goal of maximizing the expected highest probed value, i.e., $\expar{ \max_{i  \in S} X_i }$. Interestingly, due to its specific nature of randomness, this problem can be defined in two different ways, depending on whether the algorithm probes the set $S$ \emph{adaptively} or \emph{non-adaptively}. Here, ``adaptive'' means that the required algorithm is a \emph{policy} (decision tree) that sequentially decides on the next element to be probed depending on the outcomes observed up until then. In contrast, a non-adaptive algorithm would decide on the set of elements to be probed a-priori, without observing any outcomes.  Surprisingly, even for this seemingly-simple problem, identifying an optimal non-adaptive solution is known to be NP-hard~\cite{CHLLLL-NIPS16,GGM-TALG10} and it is believed that finding the optimal adaptive policy is \#P (or even PSPACE) hard~\cite{FLX-ICALP18}. As such, the natural question is: Can we efficiently compute  near-optimal adaptive and non-adaptive solutions?

As another motivating example, consider again an interviewing scenario where we may interview all candidates, but have to immediately decide upon interviewing whether to hire or reject the current candidate. Formally, this setting corresponds to the \emph{\Prophets} problem, where the value of each element $i \in [n]$ is again specified by an independent random variable $X_i$. The algorithm is required to determine a permutation $\sigma \in S_n$ in  which the outcomes $X_{\sigma(i)}$ will be observed, and a \emph{stopping time} $\tau$  to maximize the expected value of $X_{\sigma(\tau)}$. Due to a fundamental result of Hill~\cite{Hill-Journal83}, it is known that there exists an optimal adaptive policy for this problem which is in fact non-adaptive, and we can therefore assume that the optimal  permutation $\sigma^*$ is chosen a-priori. It is worth mentioning that, given a permutation, the optimal stopping time can easily be  computed by dynamic programming (see further details in Section~\ref{sec:freeOrder}). However, finding the optimal permutation has recently been proven  to be NP-hard by Agrawal et al.~\cite{ASZ-EC20}. In this context, the basic question is whether one can still obtain a non-trivial approximation.

In addition to the above-mentioned probing and prophets problems, numerous stochastic optimization problems have previously been considered, such as variants of Pandora's Box, Stochastic Matchings, and Stochastic Knapsack. Indeed, to date, constant factor approximations were attained for each of these problems, and we refer the reader to further discussion on related work in  Section~\ref{sec:related}. This current state of knowledge raises the following question:
\begin{quote}
\emph{Can we compute in polynomial time near-optimal solutions to adaptive and non-adaptive stochastic combinatorial optimization problems?}
\end{quote}

\subsection{Main results.} 

The primary contribution of this work consists of proposing a general approach to design \emph{efficient polynomial time approximation schemes} (EPTAS) for a number of stochastic combinatorial optimization problems. That is,  for any constant $\eps>0$, we obtain a $(1-\eps)$-approximation to the optimal objective value in $t(\epsilon)\cdot \poly(|{\cal I}|)$ time, where $t(\cdot)$ is a function that depends only on $\epsilon$ and $|{\cal I}|$ denotes the input length. It is worth pointing out that an EPTAS is particularly attractive from an implementation standpoint, due to ensuring that our running time dependency on the accuracy level $\eps$ is instance-independent. This property is  appealing to practitioners, who view running times such as $n^{1/\eps}$ as purely theoretical in certain settings, whereas terms of the form $(1/\eps)^{1/\eps}$ are more acceptable. In such practical settings, one would never fix $\eps = 10^{-6}$, but even $\eps = 0.1$ becomes impractical with $n^{1/\eps}$. For further discussion on EPTASes, we refer avid readers to a number of selected papers in this context \cite{HassinL04, Jansen10, FominLRS11, Bonamy+21}.

For ease of presentation, we first discuss the algorithmic implications of our framework, which will be followed by its high-level technical ideas in Section~\ref{sec:introTechniques}.

\paragraph{\Prophets.} This problem was first studied in the 1980's by Hill~\cite{Hill-Journal83}, who proved the existence of a non-adaptive optimal policy. From an algorithmic perspective, an approximation ratio of $1/2$ directly follows from the classical Prophet inequality~\cite{Krengel-Journal77,Krengel-Journal78,Samuel-Annals84}, which provides a threshold-based policy with value at least $\frac{1}{2} \cdot \expar{\max_{i \in [n]} X_i}$. In a recent work, Agrawal et al.~\cite{ASZ-EC20} obtained improved constant-factor approximations for special classes of distributions, such as when each random variable has a support size of at most two. However, for arbitrary distributions, existing approaches lose at least a constant factor in their guaranteed approximation ratio. 
As a first demonstration of the applicability of our framework, we provide  an EPTAS for this problem in Section~\ref{sec:freeOrder}.

\begin{theorem} \label{thm:Prophets}
There exists an EPTAS for the  \Prophets problem.
\end{theorem}

We remark that  Fu et al.~\cite{FLX-ICALP18} obtained an \emph{adaptive} PTAS for the \Pandora problem, which captures \Prophets as a special case (with zero costs). However, this result does not translate to a {non-adaptive} PTAS for finding a fixed permutation a-priori, and certainly not to an EPTAS. 

\paragraph{\ProbeMax.} 
The first non-trivial results for \ProbeMax, and for additional stochastic probing problems, were based on adaptivity-gap bounds. Such findings show that, up to certain constant factors, the adaptive and non-adaptive variants of a given problem are equivalent in terms of approximability~\cite{AN16,GN-IPCO13,GNS-SODA17}. Thus, specifically for \ProbeMax, one can  focus on its non-adaptive setting, $\max_{S: |S|\leq k} \expar{\max_{i \in S} X_i}$, which is a monotone submodular maximization problem, and can therefore be approximated in polynomial time~\cite{NWF-MP78}. Improving on these early results, Chen et al.~\cite{CHLLLL-NIPS16} designed a DP-based \emph{polynomial time approximation scheme} (PTAS) for non-adaptive \ProbeMax, where a $(1-\eps)$-approximation was attained in $n^{t(\eps)}$ time. It is important to emphasize that this finding does not translate to a PTAS for adaptive \ProbeMax, due to a constant factor gap between the adaptive and non-adaptive settings. In a recent breakthrough, Fu et al.~\cite{FLX-ICALP18} devised a PTAS for the adaptive \ProbeMax problem. Interestingly, their main idea is to design a policy that employs only a constant number of adaptive rounds. 
By exploiting our framework, we improve on the work of both Chen et al.~\cite{CHLLLL-NIPS16} and Fu et al.~\cite{FLX-ICALP18}, showing that the \ProbeMax problem actually admits an EPTAS. These results are established in Sections~\ref{sec:NAProbeMax} and~\ref{sec:adapProbeMax}, respectively.  

\begin{theorem} \label{thm:NAProbeMax}
There exists an EPTAS for the non-adaptive \ProbeMax problem.
\end{theorem}

\begin{theorem}\label{thm:AdapProbeMax}
There exists an EPTAS for the adaptive \ProbeMax problem.
\end{theorem}

We note that even though the adaptive \ProbeMax problem appears to be more ``difficult'' than its non-adaptive counterpart, we are not aware of any way to derive Theorem~\ref{thm:NAProbeMax} as a corollary of Theorem~\ref{thm:AdapProbeMax}. In essence, there is no natural way to transform a given decision tree into a non-adaptive solution while preserving its performance guarantee. 

\paragraph{Extensions.} 
In Section~\ref{sec:eqquivPandProphet}  we obtain an EPTAS for a  variant of the classical Pandora's Box problem~\cite{Weitzman-Econ79}. In \emph{\Pandora}, introduced by Fu et al.~\cite{FLX-ICALP18}, upon observing a random variable, one has to immediately decide  whether to select it or not. Formally, given  $n$ independent random variables $X_1, \ldots, X_n$, the outcome of each $X_i$ can be observed by paying a known cost, $c_i$. The algorithm is required to determine a permutation $\sigma \in S_n$ along with a stopping time $\tau$, so as to maximize $\expar{X_{\sigma(\tau)} - \sum_{i \leq \tau} c_{\sigma(i)}}$. 
Fu et al.~\cite{FLX-ICALP18} proposed an adaptive PTAS for this problem. Our contribution in this context is to prove that \Pandora is in fact equivalent to the \Prophets problem. Specifically, we show that the optimal solution to the former problem is a non-adaptive permutation, and that an $\alpha$-approximation for \Prophets implies an $\alpha$-approximation for \Pandora. By combining this equivalence with Theorem~\ref{thm:Prophets}, we derive the following result.

\begin{theorem} 
There exists an EPTAS for the \Pandora problem.
\end{theorem}

Finally, we show that our  framework can be leveraged to obtain analogous results for broader settings, with multiple-element selection. In particular,  in Section~\ref{sec:mmultItems}, we obtain an EPTAS for a generalization of non-adaptive \ProbeMax, where one wishes to non-adaptively select $k$ random variables, with the goal of maximizing the expected sum of the top $r$ selected variables. 

\subsection{High-level technical overview.} \label{sec:introTechniques} 

Our approach to all stochastic optimization problems mentioned above consists of presenting   reductions to the \Santa problem. In this setting, formally defined in Section~\ref{sec:santaClaus}, we are given $n$ jobs that should be assigned to $m$ machines, where each job $j\in [n]$ incurs a $D$-dimensional \emph{vector} load of $\ell_{ij} \in \reals_{+}^D$ on machine $i \in [m]$. We are additionally given vector coverage constraints ${L}_i \in  \reals_{+}^D$ for each machine $i$, and the goal is to compute an assignment in which the total vector load on each machine is at least ${L}_i$, when such an assignment exists. When $D=1$, this formulation captures the well-known Santa Claus problem~\cite{BS-STOC06}, which has been notoriously difficult, admitting constant-factor approximations only for certain special cases; see, e.g., \cite{AKS-TALG17,CCK-FOCS09,Feige-SODA08,BD-05}. In Section~\ref{sec:santaClaus}, we prove that for a constant number of machines $m$ and dimensions $D$, an EPTAS can be designed, up to slightly violating the coverage constraints.

\begin{theorem} [Informal Theorem~\ref{thm:result_santa_claus}] \label{thm:multidimInformal}
There exists an EPTAS for the \Santa problem,  up to violating the coverage constraints by a factor of $1-\epsilon$.
\end{theorem} 

At a high-level, for purposes of analysis, our approach for each stochastic optimization problem begins by breaking the optimal (non-)adaptive solution into disjoint buckets of random variables, where within any given bucket the solution's performance does not change by much (up to a $(1-\eps)$-factor). Subsequently to guessing a number of ``hyper-parameters'' that characterize each bucket, our algorithm wishes to assign the underlying random variables to the buckets defined earlier. Intuitively, the goal of these hyper-parameters is to capture crucial structural features of the optimal solution within each bucket. This reduction results in an instance of the \Santa problem where random variables can be thought of as jobs that should be assigned to machines corresponding to buckets, while simultaneously satisfying $D$ hyper-parameter constraints. 

Given the generic approach described above, the main challenge resides in defining the right bucketing and hyper-parameters, which are problem-specific decisions. It is important to mention that, besides capturing structural features of optimal solutions, hyper-parameters are required to have an additive form with respect to the assigned random variables, since machine loads are additive within the \Santa problem. In Section~\ref{sec:freeOrder}, we apply this approach to the \Prophets problem, where the application is easier than in other cases, since we only make use of a single hyper-parameter (i.e., $D=1$).  In Section~\ref{sec:NAProbeMax}, we present an application to non-adaptive \ProbeMax, where \Santa comes up in its single-machine form. Finally, in Section~\ref{sec:adapProbeMax}, we consider  adaptive \ProbeMax, where the full power of \Santa will be required.

\paragraph{Comparison to  Fu et al.~\cite{FLX-ICALP18}.}  
Technically speaking, our bucketing-with-hyper-parameters approach shares some  similarities with the block-adaptive-with-signatures approach of Fu et al. \cite{FLX-ICALP18}. There are, however,  crucial differences. First, they define a block as a subset of random variables over which adaptivity is not very helpful (up to negligible factors). Our notion of buckets is  much more general, as it also applies to problems such as \Prophets, where the optimal solution is non-adaptive. Second, and more importantly, we define buckets and hyper-parameters in order to facilitate a reduction to the \Santa problem. In contrast, Fu et al.\ guess the signature of each block very accurately (up to $1-\frac{1}{\poly(n)}$) and utilize a massive dynamic program. It is unclear whether an EPTAS can be obtained through such dynamic programs, since $\eps$-factor errors in signature-related guesses could translate to unbounded errors for the entire problem.

\subsection{Further related work.} \label{sec:related}

Evidently, in the last two decades, there has been a rapidly-growing body of work on both probing and stopping-time stochastic  optimization problems. Therefore, we mention below selected relevant papers, and refer the readers to the second author's thesis work \cite{Singla-Thesis18} and to the references therein for a fine-grained literature review.

Probing problems have become increasingly-popular in theoretical computer science, starting at the influential work of Dean et al.~\cite{DGV-FOCS04}, who studied the stochastic knapsack problem. Additional streams of literature emerged from subsequent papers related to stochastic matchings by Chen et al.~\cite{CIKMR-ICALP09}, stochastic submodular maximization by Asadpour et al.~\cite{ANS-WINE08}, and variants of Pandora's box by Kleinberg et al.~\cite{KWW-EC16} and Singla~\cite{Singla-SODA18}. These efforts resulted in constant-factor approximation algorithms, either using  implicit bounds on the adaptivity gap involved via LP (or multilinear) relaxations, or directly through explicit bounds. Further work in this context considered a wide range of problems, including knapsack~\cite{BGK-SODA11,Ma-SODA14}, orienteering~\cite{GM-ICALP09,GKNR-SODA12,BN-IPCO14}, packing integer programs~\cite{DGV-SODA05,CIKMR-ICALP09,BGLMNR-Algorithmica12}, submodular objectives~\cite{GN-IPCO13,ASW14,GNS-SODA17,BSZ-Approx19}, matchings~\cite{Adamczyk-IPL11,BGLMNR-Algorithmica12,BCNSX-APPROX15,AGM-ESA15,GKS-SODA19}, and Pandora's box models~\cite{GJSS-IPCO19,BK-EC19,GKS-SODA19,JLLS-ITCS20,BFLL-EC20,CGTTZ-FOCS20}, just to mention a few representative papers.

A concurrent research direction investigates combinatorial generalizations of the classic secretary~\cite{Dynkin-Journal63} and prophet inequality~\cite{Krengel-Journal77,Krengel-Journal78} stopping-time problems, due to their applications in algorithmic mechanism design~\cite{BIKK-SIGecom08,Lucier-SIGecom17}. For secretary problem generalizations, we refer the reader to the book chapter by Gupta and Singla~\cite{GS-arXiv20}, as the existing literature is somewhat less relevant to our current work. Hajiaghayi et al.~\cite{HKS-AAAI07} proved a prophet inequality for uniform matroids, and Alaei~\cite{Alaei-SICOMP14} obtained an asymptotically optimal $1+O(1/\sqrt{r})$ prophet inequality. A number of additional papers along these lines considered matroids~\cite{CHMS-STOC10,Yan-SODA11,KW-STOC12,EHKS-2024}, matchings~\cite{FSZ-SODA16,EFGT-EC20}, and arbitrary downward-closed constraints~\cite{Rubinstein-STOC16,RS-SODA17}. 

It is worth mentioning that all papers listed above (except for~\cite{BGK-SODA11}) lose at least a constant factor in their  approximation ratio. In contrast, there are only a handful of results for obtaining near-optimal policies. From this perspective, Bhalgat et al.~\cite{BGK-SODA11} devised a PTAS for the stochastic knapsack problem, with a $(1+\eps)$-relaxation of its packing constraint. To our knowledge, this is where the idea of block-adaptive policies was first introduced (see Section~\ref{sec:adapProbeMax}), followed by refinements due to Li and  Yuan~\cite{LiYuan-STOC13} and to Fu et al.~\cite{FLX-ICALP18} for other probing problems.  Some recent papers have also  obtained  $(1+\eps)$-approximations for special cases of prophet problems with respect to the optimal policy,   such as for 3-point distributions~\cite{ASZ-EC20} and for constant-depth laminar matroids~\cite{ANSS-EC19}.

After an initial draft of this paper has been made publicly available \cite{SS-arXiv20}, parallel independent works by Mehta et al.~\cite{MNPR-NeurIPS20} and Liu et al. \cite{LLPSS-arXiv20} have also appeared, obtaining an EPTAS for non-adaptive \ProbeMax and for \Prophets, respectively. Their techniques are very different from those presented in this paper.

\section{EPTAS for \Santa.} \label{sec:santaClaus}

In this section, we provide a formal description of the \Santa problem that lies at the heart of our algorithmic approach. With a concrete formulation in place, we show that for a fixed number of machines and  dimensions, this problem admits an efficient polynomial-time approximation scheme with a slight feasibility violation, which will  be sufficient for our purposes in subsequent sections.

\subsection{Problem description and main result.} \label{sec:santaClausProblem}

We consider a feasibility formulation of the \Santa problem. Instances of this problem consist of the following ingredients:
\begin{itemize} 
\item We are given a set of $m$ unrelated machines. Each machine $i$ is associated with an upper bound of $k_i$ on the number of jobs it can be assigned and a $D$-dimensional vector $L_i \in \reals^D_+$ that specifies a lower bound on the load vector of this machine.

\item We have a collection of $n$ jobs, each of which can be assigned to at most one machine. When job $j$ is assigned to machine $i$, we incur a $D$-dimensional load, specified by the vector $\ell_{ij} \in \reals^D_+$.
\end{itemize}
With respect to such instances, a job-to-machine assignment is defined as a function that decides for each job   which machine it is assigned to. By slightly expanding the conventional term, an assignment is allowed to leave out any given job. We say that an assignment is \emph{feasible} when each machine $i$ is assigned at most $k_i$ jobs, accumulating an overall load of at least $L_i$. Our objective is to compute a feasible assignment, or to report that the given instance is infeasible.

\paragraph{Integer programming formulation.} Moving forward, it will be instructive to express this problem via the integer program~\eqref{eqn:IP_scheduling}, whose specifics are described below. For simplicity, we assume without loss of generality that the lower bound $L_i$ on the load of each machine $i$ is a binary vector; this assumption can easily be enforced by scaling. As such, ${\cal A}_i$ will stand for the subset of dimensions where $L_i$ is \emph{active}, i.e., ${\cal A}_i = \{ d \in [D] : L_{id} = 1 \}$. We also assume that $\ell_{ijd} \in [0,1]$ for all combinations of jobs, machines, and dimensions, as any  $\ell_{ijd} > 1$ can clearly be truncated at $1$. 
\begin{equation} \label{eqn:IP_scheduling} \tag{IP}
\begin{array}{ll}
\text{(I)} \qquad {\displaystyle x_{ij} \in \{ 0,1 \}} \qquad \qquad & \forall \, i \in [m], \, j \in [n] \\
\text{(II)} \qquad {\displaystyle \sum_{i \in [m]} x_{ij} \leq 1} \qquad \qquad & \forall \, j \in [n] \\
\text{(III)} \qquad {\displaystyle \sum_{j \in [n]} x_{ij} \leq k_i} \qquad \qquad & \forall \, i \in [m] \\
\text{(IV)} \qquad {\displaystyle \sum_{j \in [n]} \ell_{ijd} x_{ij} \geq 1} \qquad \qquad & \forall \, i \in [m], \, d \in {\cal A}_i
\end{array}
\end{equation}
In this formulation, the binary variable $x_{ij}$ indicates whether job $j$ is assigned to machine $i$. Constraints~(I) and~(II) restrict the decision variables to take binary values, assigning each job to at most one machine. Constraint~(III) ensures that each machine $i$ is assigned at most $k_i$ jobs, and constraint~(IV) guarantees that the load on this machine along any active dimension $d$ is at least~$1$.

\paragraph{Main result.} As formally stated below, we prove that for a fixed number of machines $m$ and for a fixed dimension $D$, the \Santa problem admits an efficient polynomial-time approximation scheme with a slight violation of the load constraint~(IV). 

\begin{theorem} \label{thm:result_santa_claus}
When~\eqref{eqn:IP_scheduling} is feasible, we can compute a random binary vector $X\in \{0,1\}^{m \times n}$ such that,  with probability at least $1/2$: 
\begin{enumerate} 
\item Constraints~$\mathrm{(I)}$-$\mathrm{(III)}$ are satisfied.

\item Constraint~$\mathrm{(IV)}$ is \emph{$\eps$-violated}, i.e., $\sum_{j \in [n]} \ell_{ijd} X_{ij} \geq 1 - \eps$ for every $i \in [m]$ and $d \in {\cal A}_i$.  
\end{enumerate}
Our algorithm runs in $O( f( \eps, m, D ) \cdot \poly( |{\cal I}| ) )$ time, where $|{\cal I}|$ stands for the input length in its binary representation. 
\end{theorem}

\paragraph{Outline.} To establish this result, our approach consists of formulating a strengthened LP relaxation of~\eqref{eqn:IP_scheduling}, followed by devising a randomized rounding procedure, ensuring  approximate satisfaction of the load constraintst~$\mathrm{(IV)}$. The key idea is to separate ``large" and ``small" jobs, handling the former via enumeration and the latter using concentration bounds. The proof proceeds in four steps:
\begin{enumerate}
    \item We start by choosing  small constants $\delta, \eps \in (0,1)$ and partitioning the  interval $[0,1]$ into disjoint segments $[0,\delta]$, $(\delta, (1+\eps) \cdot \delta]$, $((1+\eps) \cdot \delta, (1+\eps)^2 \cdot \delta]$, and so on. This partition allows us to define a ``type'' for every job, where any two jobs $j \neq j'$ of the same type will have $\ell_{ijd}$ and $\ell_{ij'd}$ reside in the same segment for all machines and dimensions. As such, the total number of job types is only a constant, since the number of machines and dimensions is fixed.

    \item Then, for each machine $i$, we approximately guess how many jobs of each type it receives in the optimal integral solution. In particular, we guess this number exactly when machine $i$ receives fewer than $2/\delta$ jobs of some type; otherwise, we guess this quantity up to a $(1-\xi)$-factor, for some constant $\xi = \Theta(\eps)$. 

    \item This step strengthens  \eqref{eqn:IP_scheduling} by adding, for each machine $i$, a cardinality constraint on the number of jobs it receives of each type as well as on the total contribution of small jobs using the guesses mentioned in item~2.  
    
    \item Our final step employs the dependent rounding algorithm of Gandhi et al.~\cite{GandhiKPS06} on the solution of the corresponding LP relaxation. This algorithm has the property that it preserves cardinality constraints while providing a concentration bound on the total size of small jobs. These properties allow us to argue that the total load due to large jobs is (nearly) preserved due to our cardinality guesses, and that the total load due to small jobs is approximately preserved due to concentration.
\end{enumerate}

\subsection{Step 1: Job classification based on types.} \label{subsec:scheduling_classes}

We begin by introducing a way to classify jobs according to their load contributions along any dimension. To this end, let $\delta = \delta( \eps, m, D ) \in (0,1)$ be a parameter whose value will be determined later on, ensuring in particular that $\frac{ 1 }{ \delta }$ is an integer. For every machine $i \in [m]$ and for every $i$-active dimension $d \in {\cal A}_i$, we say that job $j$ is {\em $(i,d)$-large} when $\ell_{ijd} \in (\delta,1]$; otherwise, $\ell_{ijd} \in [0,\delta]$, and this job is called {\em $(i,d)$-small}. Now, let us geometrically partition the interval $[0,1]$ into pairwise-disjoint \emph{segments} ${\cal I}_0, \ldots, {\cal I}_Q$, such that ${\cal I}_0 = [0,\delta]$, ${\cal I}_1 = (\delta, (1+\eps) \cdot \delta]$, ${\cal I}_2 = ((1+\eps) \cdot \delta, (1+\eps)^2 \cdot \delta]$, and so on. Here, $Q$ is the smallest integer for which $(1+\eps)^Q \cdot \delta \geq 1$, implying that $Q = O( \frac{ 1 }{ \eps } \log \frac{ 1 }{ \delta } )$. 

For convenience, we make use of ${\cal A}$ to designate the collection of pairs $(i,d)$ for which $d$ is an $i$-active dimension, i.e., ${\cal A} = \{ (i,d): i \in [m], d \in {\cal A}_i \}$. With this notation, let us associate each job $j \in [n]$ with a $|{\cal A}|$-dimensional \emph{type vector} ${\cal T}^j = ({\cal T}^j_{id})_{ (i,d) \in {\cal A} }$, where ${\cal T}^j_{id}$ is the unique index $q$ for which $\ell_{ijd} \in {\cal I}_q$. In other words, the marginal load $\ell_{ijd}$ we would incur along dimension $d$ by assigning job $j$ to machine $i$ resides within the segment ${\cal I}_q$. Clearly, ${\cal T}^j_{id} = 0$ when this job is $(i,d)$-small, whereas ${\cal T}^j_{id} \in [Q]$ in the $(i,d)$-large case. Since different jobs may be associated with the same type vector, let ${\cal T}^{(1)}, \ldots, {\cal T}^{(T)}$ be the collection of  distinct type vectors within $\{ {\cal T}^j \}_{j \in [n]}$. Clearly, the number of type vectors is $T = O( Q^{ |{\cal A}| } ) = O( (\frac{ 1 }{ \eps } \log \frac{ 1 }{ \delta })^{ O(mD) } )$, as any such vector is $|{\cal A}|$-dimensional with each coordinate being one of the values $0, \ldots, Q$. 

Finally, to introduce our job classification, we partition the set of jobs into classes ${\cal J}_1, \ldots, {\cal J}_T$, where each class ${\cal J}_t$ stands for the collection of jobs associated with the type vector ${\cal T}^{ (t) }$, namely, ${\cal J}_t = \{ j \in [n]: {\cal T}^j = {\cal T}^{ (t) } \}$.

\subsection{Step 2: Guessing.} \label{subsec:guessing_santa}

Our intermediate goal is to define a set of valid inequalities with respect to~\eqref{eqn:IP_scheduling} that will appropriately restrict the number of jobs assigned to each machine from each of the classes ${\cal J}_1, \ldots, {\cal J}_T$. Further inequalities will ensure that, for any machine-dimension-class triplet, we are obtaining a sufficiently-large load contribution out of jobs whose individual contribution is very small by itself. Toward forming these inequalities, we remind the reader that Theorem~\ref{thm:result_santa_claus} is conditional on the feasibility of~\eqref{eqn:IP_scheduling}, in which case we use $x^{ \eqref{eqn:IP_scheduling} } \in \{0,1\}^{m \times n}$ to denote  an arbitrary feasible solution to this integer program. We proceed by efficiently guessing a number of structural properties related to this vector.

\paragraph{Number of pairwise assignments.} Our first step will focus on obtaining a pair of parameters, $N^{ \downarrow }_{it}$ and $N^{ \uparrow }_{it}$, for every machine $i \in [m]$ and class index $t \in [T]$. These quantities will be related to the unknown value of $N^{ \eqref{eqn:IP_scheduling} }_{it} = \sum_{j \in {\cal J}_t} x^{ \eqref{eqn:IP_scheduling} }_{ij}$, which is the number of class-${\cal J}_t$ jobs that are assigned by $x^{ \eqref{eqn:IP_scheduling} }$ to machine $i$. Specifically, we begin by guessing the exact value of $\min \{ N^{ \eqref{eqn:IP_scheduling} }_{it}, \frac{ 2 }{ \delta } + 1 \}$ over all machine-class pairs, meaning that there are only $O( ( \frac{ 1 }{ \delta } )^{O(mT)} )$ guesses to consider. Subsequently, $N^{ \downarrow }_{it}$ and $N^{ \uparrow }_{it}$ are defined as follows:
\begin{itemize}
    \item {\em When $N^{ \eqref{eqn:IP_scheduling} }_{it} \leq \frac{ 2 }{ \delta }$}, which is equivalent to having $\min \{ N^{ \eqref{eqn:IP_scheduling} }_{it}, \frac{ 2 }{ \delta } + 1 \} \leq \frac{ 2 }{ \delta }$, we set $N^{ \downarrow }_{it} = N^{ \uparrow }_{it} = N^{ \eqref{eqn:IP_scheduling} }_{it}$, noting that $N^{ \eqref{eqn:IP_scheduling} }_{it}$ is indeed known in this case.

    \item {\em When $N^{ \eqref{eqn:IP_scheduling} }_{it} > \frac{ 2 }{ \delta }$}, we set $N^{ \downarrow }_{it} = \frac{ 1 }{ \delta}$ and additionally guess an integer value $N^{ \uparrow }_{it} \in [(1-\xi) \cdot N^{ \eqref{eqn:IP_scheduling} }_{it}, N^{ \eqref{eqn:IP_scheduling} }_{it}]$, where $\xi = \xi( \eps, m, D ) \in (0,1)$ is a parameter whose value will be determined later on. Across all such machine-class pairs, there are $O( (\frac{ 1 }{ \xi } \log n)^{ mT } ) = O( n \cdot 2^{ O(m^2 T^2) } \cdot (\frac{ 1 }{ \xi } )^{ m T } )$ guesses to examine.
\end{itemize}

\paragraph{Load contribution of small jobs.} As a second step, we will be guessing an underestimate $S^{ \downarrow }_{idt}$ of $S^{ \eqref{eqn:IP_scheduling} }_{idt} = \sum_{j \in {\cal J}_t} \ell_{ijd} x^{ \eqref{eqn:IP_scheduling} }_{ij}$, for every machine $i \in [m]$, active dimension $d \in {\cal A}_i$, and class index $t$ with ${\cal T}^{ (t) }_{id} = 0$. To better understand this expression, recall that job $j$ is $(i,d)$-small when $\ell_{ijd} \leq \delta$, which is equivalent to having a type vector with ${\cal T}^j_{id} = 0$. Therefore, $S^{ \eqref{eqn:IP_scheduling} }_{idt}$ is exactly the total $x^{ \eqref{eqn:IP_scheduling} }$-induced load on machine $i$ along dimension $d$ due to the jobs in class ${\cal J}_t$, which are $(i,d)$-small. In this case, our guess is of the form $S^{ \downarrow }_{idt} = \lfloor S^{ \eqref{eqn:IP_scheduling} }_{idt} \rfloor_{\eps/T}$, where  $\lfloor \cdot \rfloor_{\eps/T}$ is an operator that rounds its argument down to the nearest integer multiple of $\eps/T$. Consequently, over all machine-dimension-class triplets, the number of guesses to be considered is $O( ( \frac{ T }{ \eps } )^{ O(mDT) } )$.

\subsection{Step 3: Introducing the linear program.}

To better understand the upcoming discussion, it is useful to view our linear problem as a relaxation of an edge-selection problem on a bipartite graph. One side of this graph has a separate vertex for each job $j \in [n]$. On the other side, we create a unique vertex for each machine-class pair $(i,t)$. These two vertices are connected by an edge if and only if job $j$ belongs to class ${\cal J}_t$, in which case this edge is labeled by the load vector $\ell_{ij}$. Recalling that ${\cal J}_1, \ldots, {\cal J}_T$ partitions the overall collection of jobs based on their types, it follows that job $j$ is connected to exactly one pair $(i,t)$ for each machine $i$.

Given this graph, for every job-vertex $j \in [n]$ and pair-vertex $(i,t) \in [m] \times [T]$, we create a $[0,1]$-valued decision variable $y_{(i,t),j}$. Had the latter been $\{0,1\}$-valued, it could have been perceived as indicating whether we pick the edge $(j,(i,t))$ or not, which represents the decision of assigning job $j$ to machine $i$. Due to potentially taking values strictly between $0$ and $1$, we will interpret $y_{(i,t),j}$ as the selected fraction of $(j,(i,t))$. The feasibility region of our linear program will incorporate three types of constraints:
\begin{itemize}
    \item Each job-vertex has a fractional degree of at most $1$.

    \item Each pair-vertex $(i,t)$ has $N^{ \downarrow }_{it}$ and $N^{ \uparrow }_{it}$ as lower and upper bounds on its fractional degree.

    \item For every machine $i \in [m]$, active dimension $d \in {\cal A}_i$, and class index $t$ with ${\cal T}^{ (t) }_{id} = 0$, we place a lower bound of $(1 - \xi) \cdot S^{ \downarrow }_{idt}$ on the total fractional load on machine $i$ along dimension $d$ due to the jobs in class ${\cal J}_t$.
\end{itemize}
Formally, we define the following linear feasibility problem, where in constraint~$(S_2)$ we use $\rho( j )$ to denote the unique index $t \in [T]$ for which $j \in {\cal J}_t$:  
\begin{equation} \label{eqn:linear_relax} \tag{LP}
\begin{array}{ll}
(S_1) \quad y_{(i,t),j} \in [ 0,1 ]  \qquad & \forall \, i \in [m], \, t \in [T], \, j \in {\cal J}_t  \\
(S_2) \quad {\displaystyle \sum_{i \in [m]} y_{(i,\rho( j )),j} \leq 1}  \qquad & \forall \, j \in [n] \\
(S_3) \quad {\displaystyle \sum_{j \in {\cal J}_t} y_{(i,t),j} \in [N^{ \downarrow }_{it}, N^{ \uparrow }_{it}] } \qquad & \forall \, i \in [m], \, t \in [T] \\
(S_4) \quad {\displaystyle \sum_{j \in {\cal J}_t} \ell_{ijd} y_{(i,t),j} \geq (1 - \xi) \cdot S^{ \downarrow }_{idt} \qquad} & \forall \, i \in [m], \, d \in {\cal A}_i, \, t \in [T] : {\cal T}^{ (t) }_{id} = 0
\end{array}
\end{equation}
 The next lemma, whose proof is provided in Section~\ref{subsec:proof_clm_bound}, shows that~\eqref{eqn:linear_relax} is necessarily feasible, under our running assumption that~\eqref{eqn:IP_scheduling} is feasible. 

\begin{lemma} \label{lem:lp_feasible_givenip}
\eqref{eqn:linear_relax} is feasible. 
\end{lemma}

\subsection{Step 4: The rounding algorithm.} \label{subsec:round_scheduling}

As our final step, we show that the dependent rounding framework of Gandhi et al.~\cite{GandhiKPS06} can be exploited to convert any feasible solution to the linear relaxation~\eqref{eqn:linear_relax} into a nearly-feasible assignment for our original problem~\eqref{eqn:IP_scheduling}.

\paragraph{Background.} In our  bipartite graph viewpoint, one side  has a separate vertex for each job $j \in [n]$ and the other side has a unique vertex for each machine-class pair $(i,t)$.
Translating the main result of Gandhi et al.~\cite{GandhiKPS06}  to this setting, given a vector $y$ that meets constraints~$(S_1)$-$(S_3)$, their work provides a randomized construction of a vector $Y = ( Y_{(i,t),j} )_{i \in [m], t \in [T], j \in {\cal J}_t}$ of Bernoulli random variables  that satisfies the following properties:
\begin{enumerate} 
\item[$(P_1)$] {\em Marginal distribution}: $\prpar{ Y_{(i,t),j} = 1 } = y_{(i,t),j}$, for every $i \in [m]$, $t \in [T]$, and $j \in {\cal J}_t$.

\item[$(P_2)$] {\em Degree preservation}: $Y$ satisfies constraints~$(S_2)$ and~$(S_3)$ with probability $1$.

\item[$(P_3)$] {\em Concentration inequalities}: For every pair-vertex $(i,t)$, for every $[0,1]$-valued coefficients $\{ \alpha_{(i,t),j} \}_{j \in {\cal J}_t}$, for every $\lambda \leq \expar{  \sum_{j \in {\cal J}_t} \alpha_{(i,t),j} Y_{(i,t),j} }$, and for every $\eta \in (0,1)$,
\[ \pr{ \sum_{j \in {\cal J}_t} \alpha_{(i,t),j} Y_{(i,t),j} \leq (1-\eta) \cdot \lambda } ~~\leq~~ \exp \left( - \frac{ \lambda \eta^2 }{ 2 } \right) \ . \]
\end{enumerate}

\paragraph{Calibration of parameters.} We remind the reader that our construction involves two parameters, $\delta = \delta( \eps, m, D )$ and $\xi = \xi( \eps, m, D )$, whose values have not been specified up until now. Regarding the latter, we simply set $\xi = \frac{ \eps }{ 2 }$. In contrast, the choice of $\delta$ is somewhat more cumbersome, due to technical reasons related to the appearance of $\delta$ in future concentration inequalities. To this end, as explained in Section~\ref{subsec:scheduling_classes}, we know that $T = O( (\frac{ 1 }{ \eps } \log \frac{ 1 }{ \delta })^{ O(mD) } )$, implying that there exists an efficiently-computable constant $\psi = \psi( \eps, m, D )$ satisfying $T \leq ( \frac{ \psi }{ \delta } )^{1/3}$ for all $\delta \in (0,1)$. As such, we pick $\delta = ( \frac{ \eps^3 }{ 32mD\psi^{2/3} } )^3$, assuming without loss of generality that $\frac{ 1 }{ \delta }$ is an integer.

\paragraph{The algorithm.} We first compute a feasible solution $y^*$ to the linear program~\eqref{eqn:linear_relax}, whose existence is guaranteed by Lemma~\ref{lem:lp_feasible_givenip}.  With respect to this solution, we then create the vector $Y = ( Y_{(i,t),j} )_{i \in [m], t \in [T], j \in {\cal J}_t}$ of Bernoulli variables  via the dependent rounding approach described above. These variables are translated in turn to the random vector $X = ( X_{ij} )_{ i \in [m], j \in [n] }$, where $X_{ij} = Y_{(i,\rho(j)),j}$, recalling that $\rho( j )$ denotes the unique index $t \in [T]$ for which $j \in {\cal J}_t$.

\subsection{Analysis.} \label{subsec:scheduling_analysis}

In order to derive Theorem~\ref{thm:result_santa_claus}, we first observe that $X$ satisfies constraints~(I)-(III) of the integer program \eqref{eqn:IP_scheduling} with probability 1:
\begin{itemize}
    \item {\em Constraint~(I)}: Since $Y$ is a vector of Bernoulli random variables, $X_{ij} \in \{ 0, 1 \}$ for every job $j \in [n]$ and machine $i \in [m]$.

    \item {\em Constraint~(II)}: To validate that each job $j \in [n]$ is assigned to at most one machine, note that with probability $1$, the random number of machines to which we assign this job is exactly $\sum_{i \in [m]} X_{ij} = \sum_{i \in [m]} Y_{(i,\rho(j)),j} \leq 1$, where the last inequality holds since $Y$ satisfies constraint~$(S_2)$ with probability $1$, by property~$(P_2)$.

     \item {\em Constraint~(III)}: We argue that the random number of jobs assigned to each machine $i \in [m]$ does not exceed its capacity $k_i$ with probability $1$, since
\begin{eqnarray}
\sum_{j \in [n]} X_{ij} & = & \sum_{t \in [T]} \sum_{j \in {\cal J}_t} Y_{(i,t),j} \nonumber \\
& \leq & \sum_{t \in [T]} N^{ \uparrow }_{it} \label{eqn:third_constraint_1} \\
& \leq & \sum_{t \in [T]} N^{ \eqref{eqn:IP_scheduling} }_{it} \label{eqn:third_constraint_2} \\
& = & \sum_{t \in [T]} \sum_{j \in {\cal J}_t} x^{ \eqref{eqn:IP_scheduling} }_{ij} \nonumber \\
& = & \sum_{j \in [n]}x^{ \eqref{eqn:IP_scheduling} }_{ij} \nonumber\\
& \leq & k_i \ . \label{eqn:third_constraint_3}
\end{eqnarray}
Here, inequality~\eqref{eqn:third_constraint_1} follows by noting that $\sum_{j \in {\cal J}_t} Y_{(i,t),j} \leq N^{ \uparrow }_{it}$, according to property~$(P_2)$. Inequality~\eqref{eqn:third_constraint_2} follows by definition of $N^{ \uparrow }_{it}$ in Section~\ref{subsec:guessing_santa}. Finally, inequality~\eqref{eqn:third_constraint_3} holds since $x^{ \eqref{eqn:IP_scheduling} }$ is a feasible solution to~\eqref{eqn:IP_scheduling}.
\end{itemize}

The challenging part of our analysis resides in showing that, with probability at least $1/2$, constraint~(IV) is only $O(\eps)$-violated. This result is precisely where we will be exploiting the concentration inequalities stated in property~$(P_3)$ along with the lower bounds $\{ N^{ \downarrow }_{it} \}_{i \in [m], t \in [T]}$ in constraint~$(S_3)$, which were left aside up until now.

\begin{lemma} \label{lem:fourth_constraint}
$\prpar{ X \text{ $\mathrm{violates}$ } (\mathrm{IV}) \text{ by at most  $2\eps$} } \geq 1/2$.
\end{lemma}
\proof{Proof.}
Our proof is based on proving that for every machine $i \in [m]$ and $i$-active dimension $d \in {\cal A}_i$, we have  $\prpar{ \sum_{j \in [n]} \ell_{ijd} X_{ij} < 1 - 2\eps } \leq \frac{ 1 }{2mD}$. This claim would imply that
\begin{eqnarray*}
\pr{ X \text{ $\mathrm{violates}$ } (\mathrm{IV}) \text{ by at most  $2\eps$} }  & = & \pr{ \bigcap_{i \in [m], d \in {\cal A}_i} \Big\{ \sum_{j \in [n]} \ell_{ijd} X_{ij} \geq 1 - 2\eps \Big\} } \\
& \geq & 1 - \sum_{i \in [m]} \sum_{d \in {\cal A}_i} \pr{ \sum_{j \in [n]} \ell_{ijd} X_{ij} < 1 - 2\eps } \\
& \geq &\frac{ 1 }{ 2 } \ .
\end{eqnarray*}

Focusing on one such machine-dimension pair $(i,d)$, let us decompose the random load $\sum_{j \in [n]} \ell_{ijd} X_{ij}$ into the contributions of $(i,d)$-small and $(i,d)$-large jobs as follows:
\begin{eqnarray}
\sum_{j \in [n]} \ell_{ijd} X_{ij} & = & \sum_{t \in [T]} \sum_{j \in {\cal J}_t} \ell_{ijd} Y_{(i,t),j} \nonumber \\
& = & \sum_{t \in [T]: {\cal T}^{ (t) }_{id} = 0} \underbrace{ \sum_{j \in {\cal J}_t} \ell_{ijd} Y_{(i,t),j} }_{ \mathsf{(A)}_{idt} } + \sum_{t \in [T]: {\cal T}^{ (t) }_{id} \geq 1} \underbrace{ \sum_{j \in {\cal J}_t} \ell_{ijd} Y_{(i,t),j} }_{ \mathsf{(B)}_{idt} } \ . \label{eqn:decompose_load}
\end{eqnarray}
In this decomposition, each of the terms $\mathsf{(A)}_{idt}$ aggregates load contributions due to $(i,d)$-small jobs belonging to class ${\cal J}_t$. As stated in Claim~\ref{clm:bound_first}, we relate $\mathsf{(A)}_{idt}$ to our estimate $S^{ \downarrow }_{idt}$ for the analogous quantity with respect to $x^{ \eqref{eqn:IP_scheduling} }$ via the concentration inequalities of property~$(P_3)$. In contrast, $\mathsf{(B)}_{idt}$ represents the total load contribution of $(i,d)$-large jobs belonging to class ${\cal J}_t$, shown in Claim~\ref{clm:bound_second} to nearly-match their analogous quantity with respect to $x^{ \eqref{eqn:IP_scheduling} }$ with probability $1$. For readability purposes, we provide the proofs of these claims in Section~\ref{subsec:proof_clm_bound}. 

\begin{claim} \label{clm:bound_first}
$\prpar{ \mathsf{(A)}_{idt} \geq (1-\eps) \cdot S^{ \downarrow }_{idt} } \geq 1 - \frac{ 1 }{ 2mDT }$. 
\end{claim}

\begin{claim} \label{clm:bound_second}
$\mathsf{(B)}_{idt} \geq (1 - \eps) \cdot  \min \{ \sum_{j \in {\cal J}_t} \ell_{ijd} x^{ \eqref{eqn:IP_scheduling} }_{ij}, 1 \}$, with probability $1$. 
\end{claim}

By combining these bounds and  decomposition~\eqref{eqn:decompose_load}, with probability at least $1 - \frac{ 1 }{ 2mD }$, we have 
\begin{eqnarray}
\sum_{j \in [n]} \ell_{ijd} X_{ij} & \geq & (1-\eps) \cdot \sum_{t \in [T]: {\cal T}^{ (t) }_{id} = 0} S^{ \downarrow }_{idt} + (1 - \eps) \cdot   \sum_{t \in [T]: {\cal T}^{ (t) }_{id} \geq 1} \min \left\{  \sum_{j \in {\cal J}_t} \ell_{ijd} x^{ \eqref{eqn:IP_scheduling} }_{ij}, 1 \right\} \nonumber \\
& \geq & (1 - \eps) \cdot \left( \sum_{t \in [T]: {\cal T}^{ (t) }_{id} = 0} \sum_{j \in {\cal J}_t} \ell_{ijd} x^{ \eqref{eqn:IP_scheduling} }_{ij} - \eps \right) + (1 - \eps) \cdot  \min \left\{ \sum_{t \in [T]: {\cal T}^{ (t) }_{id} \geq 1} \sum_{j \in {\cal J}_t} \ell_{ijd} x^{ \eqref{eqn:IP_scheduling} }_{ij}, 1 \right\} 
\nonumber \\
& \geq & (1 - \eps) \cdot \min \left\{ \sum_{j \in [n]} \ell_{ijd} x^{ \eqref{eqn:IP_scheduling} }_{ij}, 1 \right\} - \eps \nonumber \\
& \geq & 1 - 2 \eps \ . \label{eqn:const_IV_2}
\end{eqnarray}
Here, the second inequality holds since $S^{ \downarrow }_{idt} = \lfloor S^{ \eqref{eqn:IP_scheduling} }_{idt} \rfloor_{\eps/T} = \lfloor \sum_{j \in {\cal J}_t} \ell_{ijd} x^{ \eqref{eqn:IP_scheduling} }_{ij} \rfloor_{\eps/T}$, whereas inequality~\eqref{eqn:const_IV_2} follows by noting that $\sum_{j \in [n]} \ell_{ijd} x^{ \eqref{eqn:IP_scheduling} }_{ij} \geq 1$, as implied by the feasibility of $x^{ \eqref{eqn:IP_scheduling} }$ for~\eqref{eqn:IP_scheduling}. \halmos
\endproof

\subsection{Additional proofs.} \label{subsec:proof_clm_bound}

\proof{Proof of Lemma~\ref{lem:lp_feasible_givenip}.} 

We establish the feasibility of~\eqref{eqn:linear_relax} by constructing a candidate solution $\hat{y}$ and proving that it satisfies each of the constraints~$(S_1)$-$(S_4)$. To this end, for every pair-vertex $(i,t)$ and job-vertex $j \in {\cal J}_t$, we define
\[ \hat{y}_{(i,t),j} ~~=~~ \begin{cases}
    x^{ \eqref{eqn:IP_scheduling} }_{ij}, & \text{if } N^{ \eqref{eqn:IP_scheduling} }_{it} \leq \frac{ 2 }{ \delta } \\
    (1 - \xi) \cdot x^{ \eqref{eqn:IP_scheduling} }_{ij}, & \text{if } N^{ \eqref{eqn:IP_scheduling} }_{it} > \frac{ 2 }{ \delta }
\end{cases}\]
We proceed by arguing that $\hat{y}$ is indeed a feasible solution to~\eqref{eqn:linear_relax}:
\begin{itemize}
    \item {\em Constraint~$(S_1)$}: Since $x^{ \eqref{eqn:IP_scheduling} }_{ij} \in \{ 0,1 \}$, it follows that $\hat{y}_{(i,t),j} \leq x^{ \eqref{eqn:IP_scheduling} }_{ij} \in  [0,1]$.

    \item {\em Constraint~$(S_2)$}: For every job-vertex $j \in [n]$, we have
    $\sum_{i \in [m]} \hat{y}_{(i,\rho( j )),j} \leq \sum_{i \in [m]} x^{ \eqref{eqn:IP_scheduling} }_{ij}  \leq 1$, where the last inequality is implied by the feasibility of $x^{ \eqref{eqn:IP_scheduling} }$ for~\eqref{eqn:IP_scheduling}.

    \item {\em Constraint~$(S_3)$}: For every pair-vertex $(i,t)$ with $N^{ \eqref{eqn:IP_scheduling} }_{it} \leq \frac{ 2 }{ \delta }$, we have $\sum_{j \in {\cal J}_t} \hat{y}_{(i,t),j} = \sum_{j \in {\cal J}_t} x^{ \eqref{eqn:IP_scheduling} }_{ij} = N^{ \eqref{eqn:IP_scheduling} }_{it} \in [N^{ \downarrow }_{it}, N^{ \uparrow }_{it}]$, since $N^{ \downarrow }_{it} = N^{ \uparrow }_{it} = N^{ \eqref{eqn:IP_scheduling} }_{it}$ in this case. In the opposite scenario, where $N^{ \eqref{eqn:IP_scheduling} }_{it} > \frac{ 2 }{ \delta }$, we have
    \[ \sum_{j \in {\cal J}_t} \hat{y}_{(i,t),j} ~~=~~ (1 - \xi) \cdot \sum_{j \in {\cal J}_t} x^{ \eqref{eqn:IP_scheduling} }_{ij} ~~=~~ (1 - \xi) \cdot N^{ \eqref{eqn:IP_scheduling} }_{it} ~~\in~~ [N^{ \downarrow }_{it}, N^{ \uparrow }_{it}] \ . \]   
    To better understand the last transition, note that $(1 - \xi) \cdot N^{ \eqref{eqn:IP_scheduling} }_{it} \geq \frac{ N^{ \eqref{eqn:IP_scheduling} }_{it} }{ 2 } > \frac{ 1 }{ \delta } = N^{ \downarrow }_{it}$. In addition, since $N^{ \uparrow }_{it} \in [(1-\xi) \cdot N^{ \eqref{eqn:IP_scheduling} }_{it}, N^{ \eqref{eqn:IP_scheduling} }_{it}]$, we also have $(1 - \xi) \cdot N^{ \eqref{eqn:IP_scheduling} }_{it} \leq N^{ \uparrow }_{it}$.

    \item {\em Constraint~$(S_4)$}: For every machine $i \in [m]$, active dimension $d \in {\cal A}_i$, and class index $t$ with ${\cal T}^{ (t) }_{id} = 0$, we observe that
    \[ \sum_{j \in {\cal J}_t} \ell_{ijd} \hat{y}_{(i,t),j} ~~\geq~~ (1 - \xi) \cdot \sum_{j \in {\cal J}_t} \ell_{ijd} x^{ \eqref{eqn:IP_scheduling} }_{ij} ~~=~~ (1 - \xi) \cdot S^{ \eqref{eqn:IP_scheduling} }_{idt} ~~\geq~~ (1 - \xi) \cdot S^{ \downarrow }_{idt} \ , \]
    where the last inequality holds since $S^{ \downarrow }_{idt} = \lfloor S^{ \eqref{eqn:IP_scheduling} }_{idt} \rfloor_{\eps/T} \leq S^{ \eqref{eqn:IP_scheduling} }_{idt}$.  
\end{itemize}
\halmos
\endproof

\proof{Proof of Claim~\ref{clm:bound_first}.} 
Clearly, the claim becomes trivial when $S^{ \downarrow }_{idt} = 0$, and we therefore assume for the remainder of this proof that $S^{ \downarrow }_{idt} > 0$. In this case, since $S^{ \downarrow }_{idt} = \lfloor S^{ \eqref{eqn:IP_scheduling} }_{idt} \rfloor_{\eps/T}$, it follows that $S^{ \downarrow }_{idt} \geq \eps/T$. Let us first observe that, by property~$(P_1)$,
\[ \ex{ \mathsf{(A)}_{idt} } ~~=~~ \ex{  \sum_{j \in {\cal J}_t} \ell_{ijd} Y_{(i,t),j} } ~~=~~ \underbrace{  \sum_{j \in {\cal J}_t} \ell_{ijd} y^*_{(i,t),j} }_{ \lambda } \ . 
\]
Consequently,
\begin{eqnarray}
\pr{ \mathsf{(A)}_{idt} < (1-\eps) \cdot S^{ \downarrow }_{idt} } & \leq & \pr{ \sum_{j \in {\cal J}_t} \frac{ \ell_{ijd} }{ \delta } \cdot Y_{(i,t),j} \leq \left( 1- \frac{ \eps }{ 2 } \right) \cdot \frac{ \lambda }{ \delta } }\label{eqn:fourth_constraint_1} \\
& \leq & \exp \left( - \frac{ \lambda \eps^2 }{ 8 \delta } \right) \label{eqn:fourth_constraint_2} \\
& \leq & \exp \left( - \frac{ \eps^3 }{ 16 \delta T } \right) \label{eqn:fourth_constraint_3} \\
& \leq & \frac{ 1 }{ 2mDT } \ .  \label{eqn:fourth_constraint_4}
\end{eqnarray}
Here, inequality~\eqref{eqn:fourth_constraint_1} follows by noting that
\[ \left( 1- \frac{ \eps }{2} \right) \cdot \lambda ~~=~~ \left( 1- \frac{ \eps }{2} \right) \cdot \sum_{j \in {\cal J}_t} \ell_{ijd} y^*_{(i,t),j} ~~\geq~~ \left( 1- \frac{ \eps }{2} \right) (1 - \xi) \cdot S^{ \downarrow }_{idt} ~~\geq ~~ (1-\eps) \cdot S^{ \downarrow }_{idt} \ , \]
where the last two inequalities respectively hold since $y^*$ satisfies constraint~$(S_4)$ and since we have previously chosen $\xi = \frac{ \eps }{ 2 }$. For inequality~\eqref{eqn:fourth_constraint_2}, we instantiate property~$(P_3)$ with $\eta = \frac{ \eps }{ 2 }$; here, it is important to mention that $\frac{ \ell_{ijd} }{ \delta } \in [0,1]$, since $\ell_{ijd} \in [0,\delta]$ for every $(i,d)$-small job $j$. Inequality~\eqref{eqn:fourth_constraint_3} is obtained by noticing that $\lambda = \sum_{j \in {\cal J}_t} \ell_{ijd} y^*_{(i,t),j} \geq (1 - \xi) \cdot S^{ \downarrow }_{idt}  \geq \frac{ \eps }{ 2T }$, due to considering the case where  $S^{ \downarrow }_{idt} \geq \eps/T$. Finally, to establish inequality~\eqref{eqn:fourth_constraint_4}, note that it holds if and only if $\delta \leq \frac{ \eps^3 }{ 16 T \ln(2mDT) }$. We show that our specific choice of $\delta$ forces this equivalent inequality to be satisfied, since starting from its right-hand-side,
\[ \frac{ \eps^3 }{ 16 T \ln(2mDT) } ~~\geq~~ \frac{ \eps^3 }{ 32mD T^2 } ~~\geq~~ \frac{ \eps^3 \delta^{2/3} }{ 32mD \psi^{2/3} } ~~=~~ \delta \ . \]
Here, the first inequality holds since $\ln x \leq x$ for all $x > 0$. The second inequality follows by recalling that $T \leq ( \frac{ \psi }{ \delta } )^{1/3}$, as explained in Section~\ref{subsec:round_scheduling}. The concluding equality is precisely the definition of $\delta = ( \frac{ \eps^3 }{ 32mD\psi^{2/3} } )^3$ in rearranged form.
\halmos
\endproof

\proof{Proof of Claim~\ref{clm:bound_second}.} 

Letting $q = {\cal T}^{ (t) }_{id} \geq 1$, we have 
\begin{eqnarray}
\mathsf{(B)}_{idt} & = & \sum_{j \in {\cal J}_t} \ell_{ijd} Y_{(i,t),j} \nonumber \\
& \geq & (1+\eps)^{q-1} \cdot \delta \cdot \sum_{j \in {\cal J}_t} Y_{(i,t),j} \label{eqn:bound2_1} \\
& \geq & (1+\eps)^{q-1} \cdot \delta \cdot N^{ \downarrow }_{it} \label{eqn:bound2_3} \\
& \geq & (1+\eps)^{q-1} \cdot \delta \cdot \min \left\{ \sum_{j \in {\cal J}_t} x^{ \eqref{eqn:IP_scheduling} }_{ij}, \frac{ 1 }{ \delta } \right\} \label{eqn:bound2_4} \\
& \geq & \frac{ 1 }{ 1 + \eps } \cdot \min \left\{ (1+\eps)^q \cdot \delta \cdot \sum_{j \in {\cal J}_t} x^{ \eqref{eqn:IP_scheduling} }_{ij}, 1 \right\} \nonumber \\
& \geq & ( 1 - \eps ) \cdot \min \left\{ \sum_{j \in {\cal J}_t} \ell_{ijd} x^{ \eqref{eqn:IP_scheduling} }_{ij}, 1 \right\} \ . \label{eqn:bound2_5}
\end{eqnarray}
Here, inequalities~\eqref{eqn:bound2_1} and~\eqref{eqn:bound2_5} hold since, for any job $j \in {\cal J}_t$, having ${\cal T}^{ (t) }_{id} = q \geq 1$ implies that $\ell_{ijd} \in {\cal I}_q = ((1+\eps)^{q-1} \cdot \delta, (1+\eps)^q \cdot \delta]$. Inequality~\eqref{eqn:bound2_3} follows from property~$(P_2)$, stating that $Y$ satisfies constraint~$(S_3)$ with probability 1, meaning in particular that $\sum_{j \in {\cal J}_t} Y_{(i,t),j} \geq N^{ \downarrow }_{it}$. To obtain inequality~\eqref{eqn:bound2_4}, we observe that $N^{ \downarrow }_{it} \geq \min \{ N^{ \eqref{eqn:IP_scheduling} }_{it}, \frac{ 1 }{ \delta } \} = \min \{ \sum_{j \in {\cal J}_t} x^{ \eqref{eqn:IP_scheduling} }_{ij}, \frac{ 1 }{ \delta } \}$,  according to the definitions of $N^{ \downarrow }_{it}$ and $N^{ \eqref{eqn:IP_scheduling} }_{it}$ in Section~\ref{subsec:guessing_santa}.
\halmos
\endproof

\section{\Prophets.} \label{sec:freeOrder}

In this section, we employ our approximation scheme for \Santa to derive an EPTAS for the \Prophets problem, as stated in Theorem~\ref{thm:Prophets}. We start with this application of our approach as it only involves a reduction to the \SingleSantano problem, which is simpler to describe and analyze.

\subsection{Problem description and outline.} \label{sec:prophProbDiscrip}

In the \Prophets problem, we are given $n$ independent random variables $X_1, \ldots, X_n$. Our goal is to find a permutation $\sigma \in S_n$ by which the outcomes $X_{\sigma(\var)}$ will be observed and a stopping rule $\tau$ so as to maximize the expected value of $X_{\sigma(\tau)}$. As mentioned in Section~\ref{sec:intro}, there exists an optimal adaptive policy for this problem which is in fact non-adaptive~\cite{Hill-Journal83}, and we therefore assume that the optimal probing permutation $\sigma^*$ is chosen non-adaptively, i.e., this permutation is determined a-priori, without any dependence on the observed outcomes.

For any fixed permutation $\sigma$, let $V(\sigma)$ denote the expected value obtained by an algorithm that utilizes the optimal stopping rule on $\sigma$; it is easy to verify that this quantity can be computed in $\poly(n)$ time via dynamic programming. Specifically, with respect to the optimal permutation $\sigma^*$, we recursively define $V^*_n = \expar{\max \{X_{\sigma^*(n)},0 \}}$, $V^*_{n-1} = \expar{\max \{ X_{\sigma^*(n-1)},V^*_n \} }$, and so on, up to $V^*_1 = \expar{ \max \{ X_{\sigma^*(1)},V^*_2 \} }$. Letting $\OPT= V(\sigma^*) = V^*_1$ be the optimal expected reward, our main result is the following.

\begin{theorem} \label{thm:prophet}
Suppose that, for every random variable $X_i$ and every $v \geq 0$, we can compute $\expar{\max \{ X_i,v \}}$ in $\poly(n)$ time. Then, for any $\eps>0$, there exists a $t(\eps)\cdot \poly(n)$-time algorithm that finds a permutation $\sigma$ with expected value $V({\sigma}) \geq (1-\eps)\cdot \OPT$.
\end{theorem}

\paragraph{Outline.} Broadly speaking, our reduction begins in Sections~\ref{subsec:prophet_partition} and~\ref{subsec:prophet_guess}, where we partition the optimal permutation $\sigma^*$ into ``buckets'', distinguishing between those making ``small'' contributions to the optimal value and those potentially making ``large'' contributions via singleton variables. In Section~\ref{subsec:prophet_santa}, this bucketing scheme will allow us to rephrase the \Prophets problem via \SingleSantano terminology.  In particular, our formulation will introduce machines corresponding to buckets and jobs corresponding to random variables, with the objective of assigning variables to buckets such that every bucket receives essentially the same ``contribution'' as in the optimal permutation $\sigma^*$. 
In Section~\ref{subsec:prophet_analysis}, we will show that this approach  gaurantees a $(1-\eps)$-approximation.

\subsection{Step 1: Partitioning the optimal permutation into buckets.} \label{subsec:prophet_partition}

For simplicity of presentation, we assume without loss of generality that the inverse accuracy level $1/\eps$ is an integer, and moreover, that we have an estimate ${\cal E} \in [(1 - \eps) \cdot \OPT, \OPT]$ for the optimal expected reward. While the former assumption is trivial, the latter can be justified by considering an arbitrary permutation $\sigma$ and computing $V(\sigma)$ as explained in Section~\ref{sec:prophProbDiscrip}. The classical prophet inequality~\cite{Krengel-Journal77,Krengel-Journal78,Samuel-Annals84} shows that $V(\sigma) \geq \frac12 \expar{\max_{\var \in [n]} X_\var} \geq \frac12 \OPT$. Hence, we can simply employ the resulting algorithm for all powers of $1+\eps$ within the interval $[V(\sigma), 2\cdot V(\sigma)]$; at least one of these values corresponds to the required estimate ${\cal E}$.

As illustrated in Figure~\ref{fig:bucketing}, we say that time $t \in [n]$ is a ``jump'' if, when moving from $V^*_{t-1}$ to $V^*_t$, we cross an integer multiple of $\eps V^*_1$, meaning that the interval $[V^*_t, V^*_{t-1}]$ contains at least one such multiple. Let $\jump$ denote the number of jumps, say at times $T_1 < T_2 < \cdots < T_\jump$. Since $\OPT = V_1^* \geq \cdots \geq V_n^* \geq 0$, we clearly have $\jump \leq 1/\eps$. With respect to these jumps, for purposes of analysis, we partition the random variables $X_1, \ldots, X_n$ into $2\jump+1$ buckets (some of which could potentially be empty), alternating between ``stable'' and ``jump'' types as follows: 
\begin{itemize}
\item Stable bucket $B_1$: $X_{\sigma^*(n)},  X_{\sigma^*(n-1)}, \ldots, X_{\sigma^*(T_\jump+1)}$, i.e., all variables appearing in the permutation $\sigma^*$ after the last jump, $T_\jump$. 
    
\item Jump bucket $B_2$, consisting only of the variable $X_{\sigma^*(T_\jump)}$.
    
\item Stable bucket $B_3$: $X_{\sigma^*(T_{\jump}-1)}, \ldots,  X_{\sigma^*(T_{\jump-1}+1)}$, namely, all variables appearing strictly between the jumps $T_{\jump-1}$ and $T_\jump$. 

\item Jump bucket $B_4$, consisting only of the variable $X_{\sigma^*(T_{\jump-1})}$.

\item So on and so forth.
\end{itemize}
For each such bucket $B_\block$, we define $\baseV( B_\block )$ as the minimal value $V^*_t$ that corresponds to the variable $X_t$ appearing in the optimal permutation $\sigma^*$ right after this bucket (time-wise), meaning that $\baseV( B_\block ) = V^*_{ \max \{ t : X_{\sigma^*(t)} \in B_\block \} + 1 }$, where $V^*_{n+1} = 0$ by convention. As a result, $\baseV(B_1) = 0$, $\baseV(B_2) = V^*_{T_\jump+1}$, $\baseV(B_3) =  V^*_{T_\jump}$, and so on. 

\begin{figure}[ht]
\centering
\includegraphics[scale=0.6]{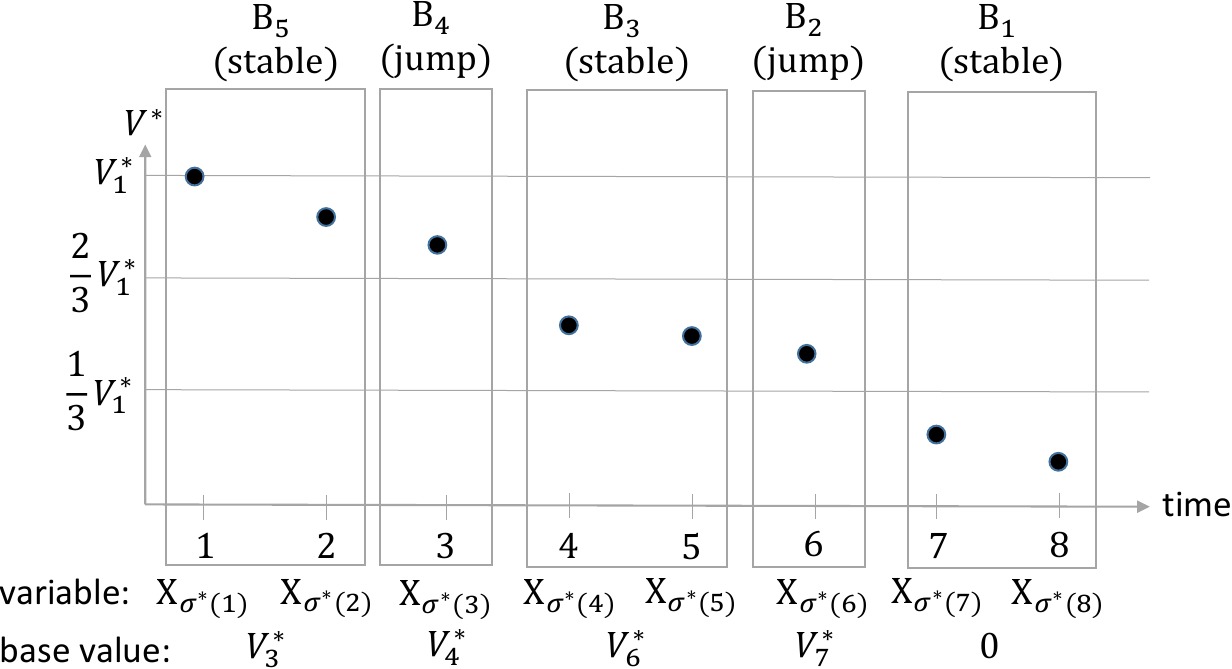} 
\caption{An illustration of our bucketing scheme.} \label{fig:bucketing}
\end{figure}

\subsection{Step 2: Guessing \baseV.} \label{subsec:prophet_guess}

We next explain why one can assume that the \baseV parameters are approximately known. To this end, we guess the \baseV of each bucket from below within an additive factor of $\eps^2 \calE$. That is, using $\baseG(B_\block)$ to denote our guess for bucket $B_\block$, we have
\begin{align} \label{eq:baseGuess}
\baseV(B_\block) - \eps^2 \calE ~~\leq ~~ \baseG(B_\block) ~~\leq~~ \baseV(B_\block) \ .
\end{align}
Since each bucket's \baseG can be enumerated over $\{ 0, \eps^2 \calE, 2\eps^2 \calE, \ldots, (1 - \eps^2) \cdot \calE \}$, there are $1/\eps^2$ options to consider, meaning that the total number of guesses is at most $(1/\eps^2)^{2\jump+1} \leq (1/\eps^2)^{(2/\eps)+1}$, which is a function of only $\eps$. Similarly, we guess the difference between the \baseV of successive buckets, $B_\block$ and $B_{\block+1}$, again within an additive factor of $\eps^2 \calE$. That is, letting $\deltaGuess(B_\block)$ designate our guess for this quantity, 
\begin{align} \label{eq:deltaGuess}
\baseV(B_{\block+1}) - \baseV(B_{\block}) - \eps^2 \calE ~~\leq~~ \deltaGuess(B_\block) 
~~\leq~~ \baseV(B_{\block+1}) - \baseV(B_\block) \ .
\end{align}
As before, each bucket's \deltaGuess can be enumerated over $\{ 0, \eps^2 \calE, 2\eps^2 \calE, \ldots, (1 - \eps^2) \cdot \calE \}$.

\subsection{Step 3: The  \SingleSantano   instance.} \label{subsec:prophet_santa}

Given $\baseG(\cdot)$ and $\deltaGuess (\cdot)$ for all buckets, we proceed by viewing \Prophets as an allocation problem, where random variables are assigned to buckets such that each bucket receives a total load of at least $\deltaGuess(\cdot)$, up to a small error. Specifically, we define the following \SingleSantano instance: 
\begin{itemize}
    \item {\em Jobs and machines}: There are $n$ jobs, corresponding to the random variables $X_1, \ldots, X_n$. In addition, we have $2\jump+1$ machines, corresponding to the buckets $B_1, \ldots, B_{2\jump+1}$.
    
    \item  {\em Assignment loads}: When job $\var \in [n]$ is assigned to machine $\block \in [2\jump+1]$, we incur a marginal load contribution of  $\ell_{\block\var} = \expar{[ X_\var - \baseG(B_\block) ]^+}$ .
    
    \item {\em Cardinality constraints}: Jump machines can be assigned at most one job, whereas stable machines can be assigned any number of jobs.
    
    \item {\em Load constraints}: Each machine $\block \in [2\jump+1]$ has a lower bound of $\deltaGuess(B_\block)$ on its total load.
\end{itemize}
Rephrasing this instance in integer programming terms, as a specialization of our generic \Santa formulation~\eqref{eqn:IP_scheduling}, we obtain:
\begin{equation} \label{eqn:IP_Santa_Prophets} \tag{IP$_{\ProphetsIP}$}
\begin{array}{ll}
(1) \quad {\displaystyle \xi_{\block\var} \in \{ 0,1 \}} \qquad \qquad & \forall \, \block \in [2\jump+1], \, \var \in [n] \\
(2) \quad {\displaystyle \sum_{\block \in [2\jump+1]} \xi_{\block\var} \leq 1} \qquad \qquad & \forall \, \var \in [n] \\
(3) \quad {\displaystyle \sum_{\var \in [n]} \xi_{\block\var} \leq 1} \qquad \qquad & \forall \, \block \in [2\jump+1] : \text{$B_\block$ jump bucket} \\
(4) \quad {\displaystyle \sum_{\var \in [n]} \ex{[ X_\var - \baseG(B_\block) ]^+} \cdot \xi_{\block\var}} \qquad \mbox{} & \forall \, \block \in [2\jump+1]  \\
{\displaystyle \qquad \qquad \qquad \qquad \qquad \geq \deltaGuess(B_\block)} 
\end{array}
\end{equation}
We observe that our bucketing partition, formally defined in Section~\ref{subsec:prophet_partition}, ensures that this problem is feasible, as shown in the next lemma.  

\begin{lemma} \label{lem:ProphetsRhoFeasibility}
\eqref{eqn:IP_Santa_Prophets} admits a feasible solution.
\end{lemma}
\proof{Proof.} We argue that our bucketing partition of $X_1, \ldots, X_n$ according to the optimal permutation $\sigma^*$, whose specifics were discussed in Section~\ref{subsec:prophet_partition}, can be exploited to define a feasible solution $\xi^*$ to~\eqref{eqn:IP_Santa_Prophets}. 
For this purpose, we distinguish between two cases, depending on bucket type:
\begin{itemize}
    \item {\em $B_\block$ is a jump bucket}: In this case, $B_\block$ consists of a single variable, say $X_{\sigma^*(t)}$, and we simply set $\xi^*_{\block \sigma^*(t)} = 1$ and $\xi^*_{\block\var} = 0$ for every other variable.
    
    \item {\em $B_\block$ is a stable bucket}: Here, $B_\block$ consists of successive variables, say $X_{\sigma^*(t)}, \ldots, X_{\sigma^*(t+k)}$. We set $\xi^*_{\block \sigma^*(t)} = \cdots = \xi^*_{\block \sigma^*(t + {k})} = 1$ and $\xi^*_{\block\var} = 0$ for every other variable.
\end{itemize}
We first notice that constraints~(1)-(3) of~\eqref{eqn:IP_Santa_Prophets} are  satisfied since $\xi^*$ is a binary vector, assigning each random variable to at most one machine, and exactly one variable to each jump machine. To show that constraint~(4) is satisfied as well, it suffices to argue that for every bucket $B_\block$ with $\deltaGuess(B_\block) > 0$,
\[ \sum_{\var \in [n]} \ex{[ X_\var - \baseG(B_\block) ]^+} \cdot \xi^*_{\block\var} ~~\geq~~ \deltaGuess(B_\block) \ . \]
 We diverge by bucket type:
\begin{itemize}
    \item {\em $B_\block$ is a jump bucket}: In this case, letting $X_{\sigma^*(t)}$ be the single variable residing in this bucket, the required lower bound is derived by noting that
    \begin{eqnarray*}
    {\textstyle \sum_{\var \in [n]}} \ex{[ X_\var - \baseG(B_\block) ]^+} \cdot \xi^*_{\block\var} & = & \ex{[ X_{\sigma^*(t)} - \baseG(B_\block) ]^+} \\
    & \geq & \ex{[ X_{\sigma^*(t)} - \baseV(B_\block) ]^+} \\
    & = & \baseV(B_{\block+1}) - \baseV(B_{\block}) \\
    & \geq & \deltaGuess(B_\block) \ ,
    \end{eqnarray*}
    where the inequalities above follow from~\eqref{eq:baseGuess} and~\eqref{eq:deltaGuess}, respectively.
    
    \item {\em $B_\block$ is a stable bucket}: Similarly, letting $X_{\sigma^*(t)}, \ldots, X_{\sigma^*(t+k)}$ be the sequence of random variables residing in $B_i$, we have
    \begin{eqnarray*}
    \textstyle \sum_{\var \in [n]} \ex{[ X_\var - \baseG(B_\block) ]^+} \cdot \xi^*_{\block\var} & = &  \textstyle \sum_{\kappa = 0}^{k} \ex{[ X_{\sigma^*(t+\kappa)} - \baseG(B_\block) ]^+} \\
    & \geq & \textstyle \sum_{\kappa = 0}^{k} \ex{[ X_{\sigma^*(t+\kappa)} - V^*_{t+\kappa+1} ]^+} \\
    & = & \baseV(B_{\block+1}) - \baseV(B_{\block}) \\
    & \geq & \deltaGuess(B_\block) \ ,
    \end{eqnarray*}
    where the first inequality holds since $V_t^* \geq \cdots \geq V_{t+k+1}^* = \baseV(B_\block) \geq \baseG(B_\block)$.   \halmos
\end{itemize}
\endproof

Consequently, by recalling that the number of machines involved is $2\jump+1 = O(1/\eps)$, Theorem~\ref{thm:result_santa_claus} immediately leads to the following result.

\begin{corollary}\label{cor:schedProphets}
There is an EPTAS for computing a cardinality-feasible variable-to-machine assignment ${x}\in \{0,1\}^{(2\jump+1) \times n}$ such that every machine $\block \in [2\jump+1]$ receives a total load of $\sum_{ \var \in [n] } \ell_{\block\var} x_{\block\var} \geq (1 - \eps ) \cdot \deltaGuess(B_\block)$.
\end{corollary}

\subsection{Final permutation and its value guarantee.} \label{subsec:prophet_analysis}

Given the assignment vector ${x}$ mentioned in  Corollary~\ref{cor:schedProphets}, we construct a probing permutation as follows. The last variables to be inspected are those assigned to the first machine (i.e., the one corresponding to $B_1$); their internal order is arbitrary. Just before them, we will inspect those assigned to the second machine (corresponding to $B_2)$, again in an arbitrary order, so on and so forth. Letting $\sigma$ be the resulting permutation, we show in the remainder of this section that $V(\sigma) \geq (1-7\eps) \cdot \OPT$, thereby establishing Theorem~\ref{thm:prophet}.

We begin by defining the sequence of values $R_1, \ldots, R_n$, attained for suffixes of the permutation $\sigma$. Letting $R_{n+1}=0$, we recursively set $R_t = \expar{\max \{ X_{\sigma(t)},R_{t+1} \} }$ for $1 \leq t \leq n$, noting that $R_1 = V(\sigma)$.  We say that our permutation $\sigma$ is \emph{behind schedule} at time $t \in [n]$ if $R_{t+1} < \baseG(B^{-1}_{ X_{\sigma(t) } })$, where $B^{-1}_{ X_{\sigma(t) } }$ designates the bucket to which $X_{\sigma(t)}$ is assigned; otherwise, $\sigma$ is \emph{ahead of schedule} at that time. Let $t_{\min}$ be the minimal time when we are ahead of schedule, i.e., satisfy $R_{t_{\min}+1} \geq \baseG(B^{-1}_{ X_{\sigma(t_{\min}) } })$. This index is well-defined, since we are always ahead of schedule at time $n$, as $R_{n+1} = 0 = \baseV(B_1) = \baseG(B_1)$, where the last equality follows from~\eqref{eq:baseGuess} and the non-negativity of $\baseG$. 

Now, noting that due to the permutation $\sigma$, we collect an expected reward of 
\begin{align}
\textstyle V( \sigma )~~ = ~~ R_1~~=~~R_{t_{\min}} + \sum_{t < t_{\min}} (R_t - R_{t+1}) \ , \label{eq:ProphBreakValue}
\end{align}
Claims~\ref{claim:free-Order-tMin} and~\ref{claim:free-Order-beforetMin} below lower-bound the last two terms.

\begin{claim} \label{claim:free-Order-tMin}
$R_{t_{\min}} \geq \baseV ( B^{-1}_{ X_{\sigma(t_{\min}) } } ) + (1 - \eps) \cdot \deltaGuess ( B^{-1}_{ X_{\sigma(t_{\min}) } } ) - 2\eps \cdot \OPT$. 
\end{claim}
\proof{Proof.}
First, we know that 
\begin{eqnarray*}
R_{t_{\min}} & = & 
\ex{ \max \{ X_{\sigma(t_{\min})}, R_{t_{\min}+1} \} } \\ 
&\geq & \ex{ \max \left\{ X_{\sigma(t_{\min})}, \baseG \left( B^{-1}_{ X_{\sigma(t_{\min}) } } \right) \right\} } \\
& = & \baseG \left( B^{-1}_{ X_{\sigma(t_{\min}) } } \right) + \ex{ \left[ X_{\sigma( t_{\min} )} - \baseG \left( B^{-1}_{ X_{\sigma(t_{\min}) } } \right) \right]^+ } \ ,
\end{eqnarray*}
where the sole inequality above holds since, by definition, we are ahead of schedule at time $t_{\min}$. The crucial observation is that the term $\expar{ [ X_{\sigma( t_{\min} )} - \baseG ( B^{-1}_{ X_{\sigma(t_{\min}) } }) ]^+ }$ represents the load contribution due to assigning job $\sigma(t_{\min})$ to the machine corresponding to bucket $B^{-1}_{ X_{\sigma(t_{\min}) } }$. When the latter is a jump bucket, our cardinality constraints ensure that job $\sigma(t_{\min})$ is the only one assigned, implying that $\expar{ [ X_{\sigma( t_{\min} )} - \baseG ( B^{-1}_{ X_{\sigma(t_{\min}) } }) ]^+ } \geq (1 - \eps) \cdot \deltaGuess(B^{-1}_{ X_{\sigma(t_{\min}) } })$ by Corollary~\ref{cor:schedProphets}. In the complementary scenario of a stable bucket, we actually have $\deltaGuess(B^{-1}_{ X_{\sigma(t_{\min}) } }) \leq \eps \cdot \OPT$, since this bucket does not cross over an integer multiple of $\eps \cdot \OPT$, by definition (see Section~\ref{subsec:prophet_partition}). In either case, we have just shown that
\begin{eqnarray} \label{eq:lowerBndRt}
 R_{t_{\min}} &\geq & \baseG \left( B^{-1}_{ X_{\sigma(t_{\min}) } } \right) + (1 - \eps) \cdot \deltaGuess \left( B^{-1}_{ X_{\sigma(t_{\min}) } } \right) - \eps \cdot \OPT \notag \\
& \geq&  \baseV \left( B^{-1}_{ X_{\sigma(t_{\min}) } } \right) + (1 - \eps) \cdot \deltaGuess \left( B^{-1}_{ X_{\sigma(t_{\min}) } } \right) - 2\eps \cdot \OPT  \ , 
\end{eqnarray}
where the second inequality holds since $\baseV(B_\block) - \eps^2 \calE \leq \baseG(B_\block)$, by inequality~\eqref{eq:baseGuess}, and since ${\cal E}\leq \OPT$.
\halmos
\endproof

\begin{claim} \label{claim:free-Order-beforetMin}
$\sum_{t < t_{\min}} (R_t - R_{t+1}) \geq (1 - \eps) \cdot \sum_{\block > B^{-1}_{ X_{\sigma(t_{\min}) } }} ( \baseV(B_{\block+1}) - \baseV(B_\block) ) - 3 \eps \cdot \OPT$. 
\end{claim}
\proof{Proof.}
We first rewrite the LHS of the desired inequality,
\begin{eqnarray*}
\sum_{t < t_{\min}} (R_t - R_{t+1}) ~~=~~ \sum_{t < t_{\min}} \left( \ex{ \max \{ X_{\sigma(t)},R_{t+1} \} } - R_{t+1} \right) 
~~=~~ \sum_{t < t_{\min}} \ex{ [X_{\sigma(t)} - R_{t+1}]^+ } \ .
\end{eqnarray*}
Recalling that $t_{\min}$ is the minimal time when we are ahead of schedule, it follows that we are behind schedule at any time $t < t_{\min}$, implying that $\expar{ [X_{\sigma(t)} - R_{t+1}]^+ } \geq \expar{ [X_{\sigma(t)} - \baseG(B^{-1}_{ X_{\sigma(t) } }) ]^+ }$. Summing this bound over all variables assigned to bucket $\block \in [2\jump+1]$, we have
\begin{eqnarray*}
\sum_{t : B^{-1}_{ X_{\sigma(t) } } = \block }\expar{ [X_{\sigma(t)} - R_{t+1}]^+ }  & \geq & \sum_{t : B^{-1}_{ X_{\sigma(t) } } = \block } \expar{ [X_{\sigma(t)} - \baseG(B_\block)]^+ } \\
& = & \sum_{ \var \in [n] } \ell_{\block\var} x_{\block\var} \\
& \geq & (1 - \eps ) \cdot \deltaGuess(B_\block),
\end{eqnarray*}
where the last inequality follows from Corollary~\ref{cor:schedProphets}. Thus, we get
\begin{eqnarray}
 \sum_{t < t_{\min}} (R_t - R_{t+1}) 
 & = & \sum_{t < t_{\min}} \ex{ [X_{\sigma(t)} - R_{t+1}]^+ } \nonumber \\
  & \geq & \sum_{\block > B^{-1}_{ X_{\sigma(t_{\min}) } }} \sum_{t : B^{-1}_{ X_{\sigma(t) } } = \block } \ex{ [X_{\sigma(t)} - R_{t+1}]^+ } \nonumber \\
  & \geq & (1 - \eps) \cdot \sum_{\block > B^{-1}_{ X_{\sigma(t_{\min}) } }} \deltaGuess(B_\block) \nonumber \\
   &\geq& (1 - \eps) \cdot \sum_{\block > B^{-1}_{ X_{\sigma(t_{\min}) } }} \Big( \baseV(B_{\block+1}) - \baseV(B_\block)  \Big) -  \eps^2 \calE \cdot (2\jump+1) \label{eqn:claim_:free-Order-beforetMin_1} \\
  &\geq&  (1 - \eps) \cdot \sum_{\block > B^{-1}_{ X_{\sigma(t_{\min}) } }} \Big( \baseV(B_{\block+1}) - \baseV(B_\block)  \Big) - 3 \eps \cdot \OPT. \label{eqn:claim_:free-Order-beforetMin_2}
 \end{eqnarray}
 Here, inequality~\eqref{eqn:claim_:free-Order-beforetMin_1} holds since $\deltaGuess(B_\block) \geq \baseV(B_{\block+1}) - \baseV(B_{\block}) - \eps^2 \calE$ by~\eqref{eq:deltaGuess}. Inequality~\eqref{eqn:claim_:free-Order-beforetMin_2} follows by recalling that $\jump \leq 1/\eps$ and that ${\cal E} \leq \OPT$.
  \halmos
\endproof

Combining the last two claims with representation~\eqref{eq:ProphBreakValue}, we conclude that our permutation $\sigma$ is guaranteed to attain an expected reward of 
 \begin{eqnarray*}
 V( \sigma )  & \geq &  \Big( \baseV \left( B^{-1}_{ X_{\sigma(t_{\min}) } } \right) + (1 - \eps) \cdot \deltaGuess \left( B^{-1}_{ X_{\sigma(t_{\min}) } } \right) -  2 \eps \cdot \OPT  \Big) \\
 && \mbox{} + \textstyle (1 - \eps) \cdot \sum_{j > B^{-1}_{ X_{\sigma(t_{\min}) } }} \Big( \baseV(B_{\block+1}) - \baseV(B_\block)  \Big) - 3 \eps \cdot \OPT \\
& \geq & (1 - 7\eps) \cdot \OPT \ .
\end{eqnarray*} 

\section{Non-Adaptive \ProbeMax.} \label{sec:NAProbeMax}

In what follows, we utilize our approximation scheme for \Santa to design an EPTAS for the non-adaptive \ProbeMax problem, as formaly stated in Theorem~\ref{thm:NAProbeMax}. Interestingly, for this application, our reduction  will create a single-machine Santa Claus instance, albeit resorting to multi-dimensional load vectors. 

\subsection{Preliminaries.} \label{sec:NAPrelims}

We remind the reader that an instance of the non-adaptive \ProbeMax problem consists of $n$ independent non-negative random variables $X_1, \ldots, X_n$. For any subset $S \subseteq [n]$, we make use of ${\cal M}(S) = \max_{j \in S} X_j$ to designate the maximum value over the sub-collection of variables $\{ X_j \}_{j \in S}$. Given an additional parameter $k$, our objective is to identify a subset $S \subseteq [n]$ of cardinality at most $k$ that maximizes $\expar{ {\cal M}(S) }$.

By referring to the technical discussion of Chen et al.~\cite[App.~C]{CHLLLL-NIPS16} in this context, the following assumptions can be made without loss of generality:
\begin{enumerate}
\item The inverse accuracy level $1 / \eps $ is an integer.

\item Letting $S^*$ be an optimal subset, an estimate ${\cal E} \in [(1 - \frac{ 1 }{ e }) \cdot \expar{ {\cal M}(S^*) }, \expar{ {\cal M}(S^*) }]$ for the optimal expected maximum is known in advance.

\item The variables $X_1, \ldots, X_n$ are defined over the same support, ${\cal V} = \big\{ 0, \eps {\cal E}, 2\eps {\cal E}, \ldots, \frac{\cal E}{ \eps } \big\}$.
\end{enumerate}
In addition, we assume without loss of generality that $\prpar{ X_j = 0 } \geq \eps$ for every $j \in [n]$. This property will be useful later on when working with logarithms of probabilities, and can be enforced by defining $\tilde{X}_j = I_j X_j$ where $I_j$ is an indicator variable taking a zero value with probability $\eps$. Letting $\tilde{\cal M}(S) = \max_{j \in S} \tilde{X}_j$, it is easy to verify that $\expar{ \tilde{\cal M}(S) } \geq (1-\eps) \cdot \expar{ {\cal M}(S) }$ for any subset $S \subseteq [n]$, and that $\tilde{\cal M}(S)$ is stochastically smaller than ${\cal M}(S)$, meaning that the expected maximum can only increase when we restore the original variables.

\subsection{Step 1: Guessing the CDF of the optimal maximum value.} \label{subsec:nonadapt_guess}

We begin by approximately guessing the cumulative distribution function of the unknown random variable ${\cal M}(S^*)$. Specifically, noting that $\prpar{ {\cal M}(S^*) \leq v }$ is monotone non-decreasing in $v$, let $v_{\mathrm{light}} \in {\cal V}$ be the maximal value $v$ for which $\prpar{ {\cal M}(S^*) \leq v } < 2\eps^2$; when no such value exists, $v_{\mathrm{light}} = -\infty$. Similarly, let $v_{\mathrm{heavy}} \in {\cal V}$ be the maximal value for which $\prpar{ {\cal M}(S^*) \leq v } < 1 - \eps^2$, where $v_{\mathrm{heavy}} = \infty$ when no such value exists. Subsequently to guessing both $v_{\mathrm{light}}$ and $v_{\mathrm{heavy}}$, we define the interval ${\cal V}_{\mathrm{critical}} = {\cal V} \cap [v_{\mathrm{light}}, v_{\mathrm{heavy}}]$ and proceed as follows: 
\begin{itemize}
\item For every value $v \in {\cal V}_{\mathrm{critical}} \setminus \{ v_{\mathrm{light}} \}$, we guess an additive over-estimate $\tilde{p}_{\leq v}$ for the probability $\prpar{ {\cal M}(S^*) \leq v }$ that satisfies
\begin{equation} \label{eqn:ineq_estimates}
\pr{ {\cal M}(S^*) \leq v } \quad \leq \quad \tilde{p}_{\leq v} \quad < \quad \pr{ {\cal M}(S^*) \leq v } + \eps^2 \ .
\end{equation}
As such, each $\tilde{p}_{\leq v}$ value can be restricted to the set $\{ 2\eps^2, 3\eps^2, \ldots, 1 - \eps^2 \}$, implying that the total number of guesses required to derive these estimates is only $O( ( \frac{ 1 }{ \eps^2 } )^{ O(| {\cal V} |) } ) = O( ( \frac{ 1 }{ \eps } )^{ O(1/\eps^2) } )$.

\item In addition, we guess a single multiplicative over-estimate $\tilde{p}_{\leq v_{\mathrm{light}}}$ for the probability $\prpar{ {\cal M}(S^*) \leq v_{\mathrm{light}} }$, such that 
    \begin{equation} \label{eqn:ineq_small_est}
    \pr{ {\cal M}(S^*) \leq v_{\mathrm{light}} } \quad \leq \quad \tilde{p}_{\leq v_{\mathrm{light}}} \quad \leq \quad 2 \cdot \pr{ {\cal M}(S^*) \leq v_{\mathrm{light}} } \ .
    \end{equation}
    For this purpose, we observe that $\prpar{ {\cal M}(S^*) \leq v_{\mathrm{light}} } \geq \prod_{j \in S^*} \prpar{ X_j = 0 } \geq \eps^k$, by our initial assumption that $\prpar{ X_j = 0 } \geq \eps$ for every $j \in [n]$, and we can restrict our estimate to powers of $2$ within the interval $[ \eps^k, 2\eps^2 ]$. Therefore, the number of required guesses is $O( k \log \frac{ 1 }{ \eps } )$.
\end{itemize}

\subsection{Why nearly matching the distribution suffices?} \label{subsec:probemax_sufficient}

Moving forward, our objective would be to identify a feasible subset $\tilde{S}\subseteq [n]$ for which the cumulative distribution function of ${\cal M}(\tilde{S})$ nearly matches that of ${\cal M}(S^*)$. To formalize this notion, we say that $\tilde{S}$ is a \emph{CDF-equivalent} subset when $\prpar{ {\cal M}(\tilde{S}) \leq v } \leq (1 + \eps^2) \cdot \tilde{p}_{\leq v} + \eps^2$ for every $v \in {\cal V}_{\mathrm{critical}}$. We first establish the performance guarantee of such subsets.

\begin{lemma} \label{lem:suff_cond}
Let $\tilde{S}\subseteq [n]$ be CDF-equivalent subset. Then, $\expar{ {\cal M}(\tilde{S}) } \geq (1 - 20\eps) \cdot \expar{ {\cal M}(S^*) }$. 
\end{lemma}
\proof{Proof.}
We begin by arguing that, for every $v \in {\cal V}$,
\begin{equation} \label{eqn:lem_suff_cond_1}
\prpar{ {\cal M}(\tilde{S}) \leq v } ~~\leq~~ (1 + \eps^2) \cdot \prpar{ {\cal M}(S^*) \leq v } + 9\eps^2 \ , 
\end{equation}
by considering three cases:
\begin{itemize}
    \item When $v \leq v_{\mathrm{light}}$, we have $\prpar{ {\cal M}(\tilde{S}) \leq v } \leq \prpar{ {\cal M}(\tilde{S}) \leq v_{\mathrm{light}} } \leq  (1 + \eps^2) \cdot \tilde{p}_{\leq v_{\mathrm{light}}} + \eps^2$, since $\tilde{S}$ is a CDF-equivalent subset. By combining this observation with the right inequality in~\eqref{eqn:ineq_small_est}, we get
      \[ \prpar{ {\cal M}(\tilde{S}) \leq v } ~~\leq ~~
      2(1 + \eps^2 )  \cdot \prpar{ {\cal M}(S^*) \leq v_{\mathrm{light}} } + \eps^2 ~~
    \leq~~ 9 \eps^2, \]
      where  the last inequality holds since $\prpar{ {\cal M}(S^*) \leq v_{\mathrm{light}} } < 2 \eps^2$ by definition of $v_{\mathrm{light}}$.
    
    \item When $v_{\mathrm{light}} < v \leq v_{\mathrm{heavy}}$, we have $\prpar{ {\cal M}(\tilde{S}) \leq v }  \leq  (1 + \eps^2) \cdot \tilde{p}_{\leq v} + \eps^2 $,  since $\tilde{S}$ is a CDF-equivalent subset. Combined with the right inequality in~\eqref{eqn:ineq_estimates}, we get 
    \[ \prpar{ {\cal M}(\tilde{S}) \leq v }
    ~~\leq~~ (1 + \eps^2) \cdot \left( \pr{ {\cal M}(S^*) \leq v } + \eps^2 \right) + \eps^2 ~~\leq~~ (1 + \eps^2) \cdot \pr{ {\cal M}(S^*) \leq v } + 3\eps^2  \ . \]
    
    \item When $v > v_{\mathrm{heavy}}$, we have
     $\prpar{ {\cal M}(\tilde{S}) \leq v } \leq 1 \leq \prpar{ {\cal M}(S^*) \leq v } + \eps^2$,
    where the second inequality holds  since $\prpar{ {\cal M}(S^*) \leq v } \geq 1 - \eps^2$ for every $v > v_{\mathrm{heavy}}$, by definition of $v_{\mathrm{heavy}}$.
\end{itemize}
Based on inequality~\eqref{eqn:lem_suff_cond_1}, noting that ${\cal M}(\tilde{S})$ and ${\cal M}(S^*)$ are both defined over the same support, ${\cal V}$, we can  relate between their expectations through the tail sum formula as follows:
\begin{eqnarray*}
\ex{ {\cal M}(\tilde{S}) } & = & \eps {\cal E} \cdot \sum_{ v \in {\cal V} } \left( 1 -  \prpar{ {\cal M}(\tilde{S}) \leq v } \right) \\
& \geq & \eps {\cal E} \cdot \sum_{ v \in {\cal V} } \left( 1 -  \left( (1 + \eps^2) \cdot \pr{ {\cal M}(S^*) \leq v } + 9\eps^2 \right) \right) \\
& \geq &  \eps {\cal E} \cdot \sum_{ v \in {\cal V} } \left( 1 -  \pr{ {\cal M}(S^*) \leq v } \right) - 10\eps^3 \calE \cdot |{\cal V}| \\
& \geq & \ex{ {\cal M}(S^*) } - 20\eps {\cal E} \\
& \geq &(1 - 20\eps) \cdot \ex{ {\cal M}(S^*) } 
\ ,
\end{eqnarray*}
where the last two inequalities hold since $|{\cal V}| = \frac{ 1 }{ \eps^2 } + 1 \leq \frac{ 2 }{ \eps^2 }$ for $\eps \in (0,1)$ and since ${\cal E} \leq \expar{ {\cal M}(S^*) }$.
\halmos
\endproof

\subsection{Step 2: The \Santa instance.} \label{subsec:probemax_santa}

We proceed by formulating an integer program that expresses CDF-related inequalities as linear covering constraints. The crucial observation is that for any set $S \subseteq [n]$ and any value $v \in {\cal V}_{\mathrm{critical}}$, we can write $\ln (\pr{ {\cal M}(S) \leq v} ) = \sum_{j \in S} \ln(\pr{X_j \leq v})$.
In turn, for every $j \in [n]$ and $v \in {\cal V}_{\mathrm{critical}}$, let us introduce the non-negative parameter $\ell_{jv} = \ln ( \frac{ 1 }{ \prpar{ X_j \leq v } } )$, noting that the latter denominator is strictly positive, by our initial assumption that $\prpar{ X_j = 0 } \geq \eps$. With this definition, consider the following feasibility-type integer problem:
\begin{equation} \label{linearize} \tag{IP$_{\mathrm{CDF}}$}
\begin{array}{ll}
(1) \quad {\displaystyle \xi_j \in \{ 0,1 \}} \qquad \qquad & \forall \, j \in [n] \\
(2) \quad {\displaystyle \sum_{j \in [n]} \xi_j \leq k} \\
(3) \quad {\displaystyle \sum_{j \in [n]} \ell_{jv} \xi_j \geq \ln \left( \frac{ 1 }{ \tilde{p}_{\leq v} } \right)} \qquad \qquad & \forall \, v \in {\cal V}_{\mathrm{critical}}
\end{array}
\end{equation}

\paragraph{Interpretation and feasibility.} A close inspection of~\eqref{linearize} reveals that we have just written an integer programming formulation of the \Santa problem on a single machine:
\begin{itemize}
    \item {\em Jobs and machines}: There are $n$ jobs, corresponding to the random variables $X_1, \ldots, X_n$. These jobs can potentially be assigned to a single machine, which captures the subset of  variables to be probed. From this perspective, the binary variable $\xi_j$ indicates whether the random variable $X_j$ is picked as part of the subset we construct.
    
    \item  {\em Assignment loads}: When job $j \in [n]$ is chosen, the marginal load contribution we incur is specified by the $|{\cal V}_{\mathrm{critical}}|$-dimensional vector $\ell_j = (\ell_{jv})$.
    
    \item {\em Cardinality constraints}: At most $k$ jobs can be assigned.
    
    \item {\em Load constraints}: Along any dimension $v \in {\cal V}_{\mathrm{critical}}$, its total load is lower bounded by $\ln( \frac{ 1 }{ \tilde{p}_{\leq v} } )$. 
\end{itemize}

The next lemma argues that the resulting \Santa instance is indeed feasible.

\begin{lemma} \label{lem:probemax_eps4}
\eqref{linearize}  admits a  feasible solution.
\end{lemma}
\proof{Proof.}
To establish the desired claim, we show that the optimal subset $S^*$ for our original \ProbeMax problem induces a feasible solution $\xi^*$ to~\eqref{linearize}. This solution is simply the incidence vector of $S^*$, meaning that $\xi^*_j = 1$ if and only if $j \in S^*$. 

First, constraints~(1) and~(2) are clearly satisfied, since $\xi^*$ is a binary vector, in which the number of assigned jobs is $\sum_{j \in [n]} \xi_j^* = |S^*| \leq k$. To verify constraint~(3), note that 
\begin{eqnarray*}
\sum_{j \in [n]} \ell_{jv} \xi_j^* &  = & \sum_{j \in S^*} \ln \left( \frac{ 1 }{ \prpar{ X_j \leq v } } \right) \\
& = & \ln \left( \frac{ 1 }{ \prod_{j \in S^*} \prpar{ X_j \leq v } } \right) \\
& = & \ln \left( \frac{ 1 }{ \prpar{ \max_{j \in S^*} X_j \leq v } } \right) \\
& = & \ln \left( \frac{ 1 }{ \prpar{ {\cal M}( S^* ) \leq v } } \right) \ ,
\end{eqnarray*}
where the third equality holds since $X_1, \ldots, X_n$ are independent. As a result, we conclude that $\xi^*$ satisfies constraint~(3) by considering two cases:
\begin{itemize}
    \item For $v \in {\cal V}_{\mathrm{critical}} \setminus \{ v_{\mathrm{light}} \}$, we have $\sum_{j \in [n]} \ell_{jv} \xi_j^* \geq \ln ( \frac{ 1 }{ \tilde{p}_{\leq v} } )$, as an immediate consequence of the first inequality in~\eqref{eqn:ineq_estimates}.  

    \item Similarly, for $v_{\mathrm{light}}$, we have $\sum_{j \in [n]} \ell_{j v_{\mathrm{light}}} \xi_j^* \geq  \ln ( \frac{ 1 }{ \tilde{p}_{\leq v_{\mathrm{light}}} } )$ as a consequence of the first inequality in~\eqref{eqn:ineq_small_est}.   	\halmos
\end{itemize}
\endproof

\subsection{Final subset and its approximation guarantee.} \label{subsec:nonadapt_final}

As an immediate byproduct of Lemma~\ref{lem:probemax_eps4}, due to considering a single-machine $O( 1 / \eps^2 )$-\Santano instance, Theorem~\ref{thm:result_santa_claus} allows us to compute a cardinality-feasible subset of variables, where the load constraint along any dimension is violated by a factor of at most $\delta$, for any fixed $\delta > 0$. This claim can be formally stated as follows.

\begin{corollary} \label{cor:non_probemax_vec}
For any fixed $\delta > 0$, there is an EPTAS for computing a vector $\xi \in \{0,1\}^n$ satisfying $\sum_{j \in [n]} \xi_j \leq k$ as well as $\sum_{j \in [n]} \ell_{jv} \xi_j \geq (1 - \delta) \cdot \ln ( \frac{ 1 }{ \tilde{p}_{\leq v} } )$, for every $v \in {\cal V}_{\mathrm{critical}}$.
\end{corollary}

Now let $\tilde{S} \subseteq [n]$ be the subset of random variable indices corresponding to the choices made by the resulting solution $\xi$, namely, $\tilde{S} = \{ j \in [n] : \xi_j = 1 \}$. This subset clearly picks at most $k$ variables, for any choice of $\delta$. We conclude our analysis by proving that $\expar{ {\cal M}(\tilde{S}) } \geq (1 - 20\eps) \cdot \expar{ {\cal M}(S^*) }$ when Corollary~\ref{cor:non_probemax_vec} is instantiated with $\delta = \frac{ \eps^3 }{ 6 }$.

For this purpose, based on Lemma~\ref{lem:suff_cond}, it suffices to show that  $\tilde{S}$ is a CDF-equivalent subset, meaning that $\prpar{ {\cal M}(\tilde{S}) \leq v } \leq (1 + \eps^2) \cdot \tilde{p}_{\leq v} + \eps^2$ for every $v \in {\cal V}_{\mathrm{critical}}$. To establish this inequality, note that
\begin{eqnarray}
\prpar{ {\cal M}(\tilde{S}) \leq v } & = &  \pr{ \max_{j \in \tilde{S}} X_j \leq v } \nonumber \\
& = &\prod_{j \in \tilde{S}} \pr{ X_j \leq v } \label{eqn:last_nonadapt_1} \\
& = &  \exp \Big\{ - \sum_{j \in \tilde{S}} \ln \Big( \frac{ 1 }{ \pr{ X_j \leq v } } \Big) \Big\} \nonumber \\
& = &\exp \Big\{ - \sum_{j \in [n]} \ell_{jv} \xi_j \Big\} \nonumber \\
& \leq &  \exp \Big\{ - \Big( 1 - \frac{ \eps^3 }{ 6 } \Big) \cdot \ln \Big( \frac{ 1 }{ \tilde{p}_{\leq v} } \Big) \Big\} \label{eqn:last_nonadapt_2} \\
& = & \tilde{p}_{\leq v}^{ 1 - \eps^3 / 6 } \ , \nonumber
\end{eqnarray}
where equality~\eqref{eqn:last_nonadapt_1} follows from the independence of $X_1, \ldots, X_n$, and inequality~\eqref{eqn:last_nonadapt_2} holds since $\sum_{j \in [n]} \ell_{jv} \xi_j \geq (1 - \frac{ \eps^3 }{ 6 }) \cdot \ln ( \frac{ 1 }{ \tilde{p}_{\leq v} } )$, via the instantiation of Corollary~\ref{cor:non_probemax_vec} with $\delta = \frac{ \eps^3 }{ 6 }$.

Now, to show that $\tilde{p}_{\leq v}^{ 1 - \eps^3 / 6 } \leq (1 + \eps^2) \cdot \tilde{p}_{\leq v} + \eps^2$, note that the differentiable function $x \mapsto x^{ 1 - \eps^3 / 6 }$ is concave over $[0,\infty)$. Therefore, by applying the gradient inequality, $\tilde{p}_{\leq v}^{ 1 - \eps^3 / 6 } \leq x^{ 1 - \eps^3 / 6 } + (1 - \frac{ \eps^3 }{ 6 }) \cdot x^{ - \eps^3 / 6 } \cdot (p-x)$ for every $x \geq 0$. Plugging in $x = \eps^3$, we get
\[ \tilde{p}_{\leq v}^{ 1 - \eps^3 / 6 } ~~\leq~~ \eps^{ 3(1 - \eps^3 / 6) } + \left( 1 - \frac{ \eps^3 }{ 6 } \right) \cdot \eps^{ - \eps^3 / 2 } \cdot \left( \tilde{p}_{\leq v} - \eps^3 \right) 
~~\leq~~ \eps^2 + \eps^{ - \eps^3 / 2 } \cdot \tilde{p}_{\leq v} ~~\leq ~~ (1 + \eps^2) \cdot \tilde{p}_{\leq v} + \eps^2 \ , \]
where the last inequality holds since $\eps^{ - \eps^3 / 2 } \leq 1 + \eps^2$ for $\eps \geq 0$.

\section{Adaptive \ProbeMax.} \label{sec:adapProbeMax}

In what follows, we employ our algorithmic framework to design an EPTAS for the adaptive \ProbeMax problem, as formally stated in Theorem~\ref{thm:AdapProbeMax}.

\subsection{Preliminaries.} \label{subsec:nonadapt_prelim}

Similarly to its non-adaptive variant, an instance of the adaptive \ProbeMax problem consists of $n$ independent non-negative random variables $X_1, \ldots, X_n$. For any subset $S \subseteq [n]$, we use ${\cal M}(S) = \max_{j \in S} X_j$ to designate the maximum value over the subcollection of variables $\{ X_j \}_{j \in S}$. Given a cardinality bound of $k$ on the number of variables to be chosen, our objective is to \emph{adaptively} identify a subset $S \subseteq [n]$ that maximizes $\expar{ {\cal M}(S) }$. Here, ``adaptive'' implies that the required algorithm is a \emph{policy} (i.e., decision tree) for sequentially deciding on the next variable to be probed, depending on the outcomes observed up until that point in time.

\paragraph{Block-adaptive policies.} The reduction we propose will exploit the notion of a block-adaptive policy, as defined by Fu et al.~\cite{FLX-ICALP18}. Formally, such policies correspond to a decision tree $\calT$ where each node $B$ represents a \emph{block} that designates a subset $\calT_B \subseteq [n]$ of variable indices. This tree translates to an adaptive  policy where, starting at the root, we simultaneously probe all random variables in the current block $B$, and depending on the outcome $\calM(\calT_B)$, decide which child block to probe next. A block adaptive policy $\calT$ is said to be \emph{feasible} when on every root-leaf path: (a)~Each variable appears at most once; and (b)~The total number of variables is at most $k$. 

\paragraph{Assumptions.} Summarizing previous work in this context by Fu et al.~\cite{FLX-ICALP18}, we make the following assumptions without loss of generality:
\begin{enumerate} 
\item The inverse accuracy level $1 / \eps $ is an integer.

\item Letting $\OPT$ denote the expected value of an optimal policy, an estimate ${\cal E} \in [(1 -\eps) \cdot \OPT, \OPT]$  is known in advance.

\item \label{assum:SmallSupport} The variables $X_1, \ldots, X_n$ are defined over the same support, ${\cal V} = \{ 0, \eps {\cal E}, 2\eps {\cal E}, \ldots, \frac{\cal E}{ \eps } \}$.

\item \label{assum:ConstDepth} There exists a feasible block-adaptive policy $\calT^*$ with an objective value of $V(\calT^*) \geq (1-\eps) \cdot \OPT$, and with $O(1/\eps^3)$ blocks on any root-leaf path.
\end{enumerate}

An immediate implication of Assumptions~\ref{assum:SmallSupport} and~\ref{assum:ConstDepth} is that the number of nodes in the decision tree $\calT^*$ is only a function of $\eps$. Thus, at the expense of increasing the number of nodes by an $\eps$-dependent factor, we  assume without loss of generality that each internal node in $\calT^*$ has an out-degree of exactly $|{\cal V}|$, meaning that it always branches for different maximum block values. 

\paragraph{Notation.} Moving forward, we denote by $V(\calT)$ the expected value of a feasible block-adaptive policy $\calT$. 
For any block $B$, we  use  $I_B$ to denote the  maximal value observed across all probed random variables leading to block $B$, with the convention that $I_B = 0$ for the root block. Finally, $\calM_B(\calT)$  stands for the (random) maximum value observed immediately after probing block $B$, i.e., $\calM_B(\calT) = \max \{ I_B , \calM(\calT_B)\}$.

\subsection{Step 1: Guessing the graph structure and block-configuration CDFs.}

We initially guess the graph structure (i.e., nodes and arcs) of the decision tree $\calT^*$ as well as the $I_B^*$-value of each block $B$. Following the discussion in Section~\ref{subsec:nonadapt_prelim}, since $\calT^*$ has $O(1/\eps^3)$ blocks on any root-leaf path, and since each internal node has an out-degree of $|{\cal V}|$, the number of required guesses for how $\calT^*$ is graphically structured  depends only on $\eps$. The same goes for the $I_B^*$-value of each block. That said, unlike the non-adaptive setting, we now run into two complicating features: (a)~While each variable may appear within multiple blocks, it is restricted to at most one appearance on any root-leaf path; and (b)~The number of probed variables on every root-leaf path is upper-bounded by $k$. 

To overcome these difficulties, we define the notion of a ``configuration'', corresponding to a structured subset of blocks. Formally, any configuration $C$ is a subset of blocks in $\calT^*$ that contains at most one block on each root-leaf path. Intuitively, a block-adaptive policy $\calT$ can be thought of as ``assigning'' each random variable $X_j$ to configuration $C$ when this variable appears in exactly the blocks belonging to $C$, i.e., $j \in \calT_B$ if and only if $B \in C$. As such, our requirement that each variable may appear at most once on any root-leaf path is precisely the reason why every configuration may contain at most one block on each root-leaf path. With these definitions, we can further partition the subset $\calT_B^*$ of variables appearing in block $B$ according to their configuration, $\{ \calT_{B,C}^* \}_{C \in \calC}$, where $\calT_{B,C}^*$ stands for the subset of variables in $B$ that are assigned by $\calT^*$ to configuration $C$. We make use of $\calC$ to designate the set of all possible configurations; clearly, $| \calC |$ only depends on $\epsilon$.

Next, along the lines of Section~\ref{subsec:nonadapt_guess}, we approximately guess the cumulative distribution function of the unknown random variable $\calM_{B,C}(\calT^*) = \max \{ I_B^* , \calM(\calT^*_{B,C})\}$ for every block $B$ and configuration $C$. Note that $\calM_{B,C}(\calT^*)$ represents the maximal value observed across all probed random variables leading to block $B$ as well as across those residing in $\calT^*_{B,C}$. To implement this guessing step, let $v^{B,C}_{\mathrm{heavy}} \in {\cal V}$ be the maximal value for which $\prpar{ {\calM_{B,C}(\calT^*)} \leq v } < 1 - \delta^2$, where $\delta = \delta(\eps)$ is a parameter whose value will be determined later on; by convention, $v^{B,C}_{\mathrm{heavy}} = \infty$ when no such value exists. Having guessed $v^{B,C}_{\mathrm{heavy}}$, we define the interval ${\cal V}^{B,C}_{\mathrm{critical}} = {\cal V} \cap [0, v^{B,C}_{\mathrm{heavy}}]$, and guess for every value $v \in {\cal V}^{B,C}_{\mathrm{critical}}$, an estimate $\tilde{p}^{B,C}_{\leq v}$ for the probability $\prpar{ \calM_{B,C}(\calT^*) \leq v }$ that satisfies
\begin{equation} \label{eqn:ineq_estimates_Adap}
\pr{ \calM_{B,C}(\calT^*) \leq v } \quad \leq \quad \tilde{p}^{B,C}_{\leq v} \quad < \quad \pr{ \calM_{B,C}(\calT^*)  \leq v } + 2\delta^2 \ . 
\end{equation}
Each such estimate $\tilde{p}^{B,C}_{\leq v}$ is restricted to the set $\{ 2\delta^2, 3\delta^2, \ldots, 1 - \delta^2 \}$, implying that the total number of guesses required to derive these estimates is only a function of $\eps$. 

\subsection{Why nearly matching the distribution of a policy suffices?}

In order to motivate subsequent steps, we point out that our objective would be to identify a feasible block-adaptive policy  $\tcalT$, having precisely the same graph structure and $I_B^*$-values as $\calT^*$, while ensuring that the cumulative distribution function of $\calM_{B,C}(\tcalT)$ nearly matches that of $\calM_{B,C}(\calT^*)$, for every block $B$ and configuration $C$. Put in concrete terms, we say that $\tcalT$ is CDF-equivalent to $\calT^*$ when it has identical graph structure and $I_B^*$-values, and moreover, when $\prpar{ \calM_{B,C}(\tcalT) \leq v } \leq (1 + \delta^2) \cdot \tilde{p}^{B,C}_{\leq v} + \delta^2$, for every block $B$, configuration $C$, and value $v \in {\cal V}^{B,C}_{\mathrm{critical}}$. The performance guarantee of such policies is stated in the next claim, whose proof is omitted, due to being nearly-identical to that of Lemma~\ref{lem:suff_cond}, with the addition of a straightforward recursion. 

\begin{lemma} \label{lem:suff_cond_Adap}
Suppose that the block-adaptive policy $\tcalT$ is CDF-equivalent to $\calT^*$. Then, $V(\tcalT) \geq (1 - g(\delta, \eps)) \cdot V(\calT^*)$, for some function $g(\cdot,\cdot)$ that depends only on $\delta$ and $\eps$. 
\end{lemma}

\subsection{Step 2: The \Santa instance.}

\paragraph{Integer program.} We next formulate an integer program that expresses CDF-equivalence requirements as linear covering constraints. To this end, for each configuration $C \in \calC$, let $k_C^*$ be the number of random variables that appear in configuration $C$ with respect to the policy $\calT^*$. To define appropriate cardinality constraints, we start by estimating $k_C^*$ for every configuration $C \in \calC$. In particular, we guess an additive over-estimate $\tilde{k}_C \in \{0 , \gamma  k, 2\gamma  k, \ldots, \lceil\frac{1}{\gamma}\rceil \cdot \gamma  k \}$ that satisfies
\begin{align} \label{eq:guessK}
k_C^* \quad \leq \quad  \tilde{k}_C \quad \leq \quad k_C^* + \gamma  k,
\end{align}
where $\gamma = \gamma(\eps)$ is a parameter whose value will be decided later on. Due to having only $O(1/\gamma)$ potential values for each such $\tilde{k}_C$, the overall number of required guesses is $O( 1/\gamma^{ |\calC| } )$, which is once again a function of $\eps$ and nothing more.

Given these quantities, for every variable index $j \in [n]$, configuration $C$, block $B$, and  value $v \in {\cal V}^{B,C}_{\mathrm{critical}}$, let us define a non-negative parameter $\ell_{CjBv} = \ln ( \frac{ 1 }{ \prpar{ \max \{ I_B^*, X_j \} \leq v } } )$,  with the convention that $\ell_{CjBv} = \infty$ when $\prpar{ \max \{ I_B^*, X_j \} \leq v } = 0$. We proceed by considering the following feasibility integer problem:
\begin{equation} \label{linearizeAdap} \tag{IP$_{\mathrm{CDF}}^{\calT^*}$}
\begin{array}{ll}
(1) \quad {\displaystyle \xi_{Cj}\in \{ 0,1 \}} \qquad \qquad & \forall \, j \in [n], \, C \in \calC \\
(2) \quad {\displaystyle \sum_{j \in [n]} \xi_{Cj} \leq \tilde{k}_C}  & \forall \, C \in \calC \\
(3) \quad {\displaystyle \sum_{j \in [n]}  \ell_{CjBv} ~\xi_{Cj} \geq  \ln \Big( \frac{ 1 }{ \tilde{p}^{B,C}_{\leq v} } \Big)} \qquad \qquad & \forall \, C \in \calC, \, B, \, v \in {\cal V}^{B,C}_{\mathrm{critical}}
\end{array}
\end{equation}

\paragraph{Interpretation and feasibility.} In spite of the cumbersome notation involved, it is not difficult to verify that~\eqref{linearizeAdap} is in fact an integer programming formulation of the next \Santa problem:
\begin{itemize}
    \item {\em Jobs and machines}: There are $n$ jobs, corresponding to the random variables $X_1, \ldots, X_n$, whereas each configuration $C$ is represented by a distinct machine. This way, the binary variable $\xi_{Cj}$ indicates whether the random variable $X_j$ is assigned in configuration $C$ within our block-adaptive policy.
    
    \item  {\em Assignment loads}: When job $j \in [n]$ is assigned to machine $C$, the marginal load contribution we incur along any block-value dimension $(B,v)$ is specified by $\ell_{CjBv}$.
    
    \item {\em Cardinality constraints}: At most $\tilde{k}_C$ jobs can be assigned to each machine $C$.

    \item {\em Load constraints}: For any machine $C$, we have a lower bound of $\ln ( 1 / \tilde{p}^{B,C}_{\leq v} )$ on the total load along any block-value dimension $(B,v)$.  
\end{itemize}

The next lemma argues that the \Santa instance we have just constructed is feasible. Noting that the arguments involved are nearly-identical to those of Lemma~\ref{lem:probemax_eps4}, we omit the proof.  

\begin{lemma} \label{lem:probemax_eps4Adap}
The integer program \eqref{linearizeAdap}  admits a  feasible solution. 
\end{lemma}

\subsection{Final block-adaptive policy.} 

As an immediate consequence, we can employ Theorem~\ref{thm:result_santa_claus} to compute a cardinality-feasible solution, where the load constraint along any dimension is violated by a factor of at most $\varphi$, for any fixed $\varphi > 0$. This claim can be formally stated as follows.

\begin{corollary} 
There is an EPTAS for computing a vector $\xi \in \{0,1\}^{\calC \times n}$ satisfying:
\begin{enumerate}
    \item $\sum_{j \in [n]} \xi_{Cj} \leq \tilde{k}_C$ for every configuration $C$.
    
    \item $\sum_{j \in [n]} \ell_{CjBv} \geq (1 - \varphi) \cdot \ln ( 1 / \tilde{p}^{B,C}_{\leq v} )$, for every configuration $C$, block $B$, and value $v \in {\cal V}^{B,C}_{\mathrm{critical}}$.
\end{enumerate}
\end{corollary}

Now, as we have written an inexact cardinality constraint for each configuration $C$, where the estimate $\tilde{k}_C$ appears in place of the unknown quantity $k_C^*$, the resulting solution $\xi$ cannot be directly translated to a choice of variables for each block, as we may end up with more than $k$ variables on various root-leaf paths.  We therefore create a random block-adaptive policy $\tcalT$, corresponding to a random sparsification of the choices made by $\xi$. Formally, we first draw a random set of indices $R \subseteq [n]$, to which each $j \in [n]$ is independently picked with probability $1 - \phi$, where $\phi = \phi(\eps)$ depends on $\eps$ and nothing more. Out of this set, within each block $B$, we place the collection of variables that were assigned by $\xi$ to a configuration containing $B$, meaning that its resulting set of variables is  $\tcalT_B = \{ j \in R : \sum_{C \ni B} \xi_{Cj} = 1 \}$; variables outside of $R$ are not assigned to any block. 

It is not difficult to verify, via standard Chernoff bounds for independent Bernoulli variables~\cite{Chernoff-52}, that the parameter $\phi$ can be chosen (as a function of $\eps$) to guarantee, with constant probability, at most $k$ variables on every root-leaf path. Moreover, Theorem~\ref{thm:AdapProbeMax} can now be derived through the sufficient near-optimality condition in Lemma~\ref{lem:suff_cond_Adap}, proving that $\tcalT$ is CDF-equivalent to $\calT^*$ along the lines of our analysis for the non-adaptive setting in Section~\ref{subsec:nonadapt_final}.

\section{Extensions.} \label{sec:exxtensions}

In this section, we prove that our EPTAS extends to the \Pandora problem, as well as to a generalization of non-adaptive \ProbeMax where one wishes to maximize the expected sum of the top $r$ selected variables.

\subsection{Equivalence of \Pandora and \Prophets.} \label{sec:eqquivPandProphet}

In what follows, we consider the \Pandora problem, showing how to obtain an EPTAS through a reduction to \Prophets (see Section~\ref{sec:freeOrder}). In this setting, we are given $n$ independent random variables $X_1, \ldots, X_n$, as well as a cost $c_i \in \reals_+$ for probing each $X_i$. Our goal is to find a probing permutation $\sigma$ and a stopping rule $\tau$  to maximize $\expar{ X_{\sigma(\tau)} - \sum_{t \leq \tau } c_{\sigma(t)} }$. In other words, the algorithm wishes to maximize the difference between the selected random variable's value $X_{\sigma(\tau)}$ and the sum of all probing costs $\sum_{t \leq \tau } c_{\sigma(t)} $. When our algorithm is allowed to claim value for the highest random variable ever probed, we end up with the classical Pandora's box problem~\cite{Weitzman-Econ79,KWW-EC16}.

Clearly, \Pandora generalizes the \Prophets problem, as the latter corresponds to the special case of zero probing costs. 
Interestingly, similar to \Prophets, it is not difficult to verify that there exists an optimal policy where the probing permutation $\sigma$ is selected non-adaptively, i.e., without any dependence on the outcomes of the probed random variables. 
Somewhat informally, this claim follows by examining the recursive equations of a natural DP-formulation, and discovering that whenever an optimal policy decides to keep probing, the specific outcomes observed thus far are irrelevant for all future actions, as we can no longer pick any of these values. A  formal proof follows Hill's analysis~\cite{Hill-Journal83}.

Our main result in this context relates between these two problems in the opposite direction, by presenting an approximation-preserving reduction.

\begin{theorem} \label{thm:equivalence}
Given oracle access to an $\alpha$-approximation for the \Prophets problem, there exists a polynomial time $\alpha$-approximation for \Pandora.
\end{theorem}
\proof{Proof.}
Our proof exploits Weitzman's index for the classical Pandora's box problem~\cite{Weitzman-Econ79}, letting $\kappa_i$  be the unique solution to $\expar{(X_i-\kappa_i)^+}=c_i$.\footnote{Here, we assume that all random variables are  continuous. Standard arguments imply that any random variable  can be approximated to  arbitrary precision by a continuous random variable. In the corner case of  $c_i=0$, we take $\kappa_i = \infty$.} With respect to this index, we define the random variables $Y_i = \min\{X_i, \kappa_i\}$. 

Now consider the permutation $\sigma$ returned by our $\alpha$-approximate \Prophets oracle, when applied to the random variables $Y_1, \ldots, Y_n$. Using $\sigma^*$ to denote an optimal permutation for this instance, we have $V_{{Y}}(\sigma) \geq \alpha \cdot V_{{Y}}(\sigma^*)$. In Claim~\ref{claim:PandCommitUpper}, we show that $V_{{Y}}(\sigma^*) \geq \OPT({X},{c})$, where $\OPT({X},{c})$ designates the expected value of an optimal solution to \Pandora. 

\begin{claim} \label{claim:PandCommitUpper}
$V_{{Y}}(\sigma^*) \geq \OPT({X},{c})$.
\end{claim}

\proof{Proof.}
Consider the optimal permutation for the \Pandora instance, attaining an objective value of $\OPT({X},{c})$. By reindexing, we can assume without loss of generality that this permutation is $X_1, \ldots, X_n$. Moreover, similar to the \Prophets problem, given the latter probing order, it is easy to verify that the optimal stopping rule is obtained by computing a non-increasing sequence of thresholds $t_1 \geq \cdots \geq t_n$, where $t_i = \OPT((X_{i+1},\ldots,X_n), (c_{i+1},\ldots, c_n))$, stopping at the first $X_i \geq t_i$. 

We first observe that $t_i \leq \kappa_i$ for all $i \in [n]$. Otherwise, $t_i>\kappa_i$, and upon reaching $X_i$ the difference in the optimal policy's utility between probing and skipping $X_i$ is given by 
\begin{eqnarray*}
\ex{-c_i + X_i\cdot \one_{X_i\geq t_i} + t_i \cdot\one_{X_i < t_i} } - t_i &  = & -c_i + \ex{ (X_i - t_i ) \cdot\one_{X_i \geq t_i} } \\
& = & -c_i +\ex{(X_i-t_i)^+ } \\
& \leq & -c_i +\ex{(X_i-\kappa_i)^+ } \\
& = & 0 \ .
\end{eqnarray*}
Thus, the optimal policy can only perform better by skipping $X_i$.

Having shown that $t_i \leq \kappa_i$, consider a policy for the \Prophets problem with respect to the $Y$-variables, that probes according to the sequence $Y_1, \ldots, Y_n$ and selects the first $Y_i \geq t_i$. We claim that the expected value of this policy is exactly $\OPT({X},{c})$, which implies $V_{{Y}}(\sigma^*) \geq \OPT({X},{c})$. To this end, we first argue that this policy stops at exactly the same index $i$ as the optimal policy for \Pandora with respect to the $X$-variables, since $t_i \leq \kappa_i$ implies that $X_i \geq t_i$ if and only if $Y_i \geq t_i$. Moreover, the expected utility is identical since, although for \Pandora we gain an additional $X_i-\kappa_i$ value over \Prophets whenever $X_i>\kappa_i$, we concurrently pay an extra cost of $c_i$. These quantities are equal in expectation, since $\expar{(X_i-\kappa_i)^+}=c_i$, and hence the same expected utility is obtained for both problems.
\halmos
\endproof

To conclude the proof, we explain how to convert the permutation $\sigma$ to a solution for the \Pandora problem with expected utility at least $V_{{Y}}(\sigma)$. 

\begin{claim} \label{claim:PandCommitLower}
The permutation $\sigma$ can be converted to a solution for the \Pandora problem with expected utility at least $V_{{Y}}(\sigma)$.
\end{claim}

\proof{Proof.}
Again by reindexing, we can assume without loss of generality that $\sigma$ corresponds to the permutation $Y_1, \ldots, Y_n$. Moreover, given this order, the optimal stopping rule consists of computing a non-increasing sequence of thresholds $t_1 \geq \cdots \geq t_n$, stopping at the first $Y_i \geq t_i$. Finally, we can assume that $t_i \leq \kappa_i$, since otherwise, the optimal policy will never stop at $Y_i = \min\{X_i, \kappa_i\}$, meaning that such variables can be ignored (skipped) while preserving the utility $V_{{Y}}(\sigma)$.

Now consider a policy for the \Pandora problem, where we probe according to the permutation $X_1, \ldots, X_n$, and select the first $X_i \geq t_i$. We claim that this policy has an expected utility of at least $V_{{Y}}(\sigma)$. First, note that this policy stops at exactly the same index $i$ as our policy for \Prophets with respect to the $Y$-variables, since $t_i \leq \kappa_i$ implies that $X_i \geq t_i$ if and only if $Y_i \geq t_i$, exactly as in the proof of Claim~\ref{claim:PandCommitUpper}. To argue about costs, note that whenever $X_i \geq \kappa_i$, although both policies stop at $i$, the difference in the obtained values is $X_i - Y_i = X_i - \kappa_i$ whereas the cost difference is $c_i$.  Since $c_i =\expar{(X_i - \kappa_i)^+} = \expar{(X_i - \kappa_i)\cdot \one_{X_i \geq \kappa_i}}$, the same expected utility is obtained for both problems.
\halmos
\endproof
This completes the proof of Theorem~\ref{thm:equivalence}.
\halmos
\endproof

\subsection{Selecting multiple elements.} \label{sec:mmultItems}

In what follows, we explain how our algorithmic framework can be adapted to obtain an EPTAS for the non-adaptive  \ProbeTopr problem, which captures non-adaptive \ProbeMax for $r=1$. In this generalization, the goal is to non-adaptively select a set $S\subseteq [n]$ of at most $k$ random variables that jointly maximize the expected sum of their top $r$ values, to which we refer as ${\cal M}_r(S)$. At a high-level, our approach differentiates between two parametric regimes. When $r \leq 1/\eps^3$, we show that single-variable \ProbeMax techniques can be extended to multiple-variable selection, primarily by considering finer discretizations. In the complementary regime of $r > 1/\eps^3$, we argue that due to falling into a ``large-budget'' setting, concentration bounds can be employed to round a natural LP-relaxation with only an $O(\eps)$-factor loss in optimality.

\begin{lemma}
When $r \leq 1/\eps^3$, there exists an EPTAS for non-adaptive  \ProbeTopr.
\end{lemma}

\proof{Proof overview.}
We consider two cases depending on the magnitude of $k$.

{\em Case 1: $k > 1/\eps^4$}. Noting that $k > r / \eps$,  we reduce  \ProbeTopr to its single-variable counterpart, non-adaptive \ProbeMax, for which Theorem~\ref{thm:NAProbeMax} provides an EPTAS. The idea is to randomly partition the variables $X_1, \ldots, X_n$ into $r/\eps$ parts $P_1, \ldots, P_{r/\eps}$, by independently and uniformly assigning each variable to one of these parts. Out of each part, we select a subset of size $\frac{ \eps }{ r } k$ to maximize the expected value of its maximum via our EPTAS for non-adaptive \ProbeMax. Letting $S_1, \ldots, S_{r/\eps}$ be the chosen subsets, basic occupancy-related concentration bounds show that $\calM_r(S_1 \cup \cdots \cup S_{r/\eps}) = (1-O(\eps)) \cdot \calM_r(S^*)$.
    
{\em Case 2: $k < 1/\eps^4$}. Due to having an $\eps$-dependent number of random variables to select, our algorithm can approximately guess  the distribution of each random variable with respect to the optimal set $S^*$. Formally, along the lines of Section~\ref{sec:NAPrelims}, we make the following assumptions without loss of generality: 
\begin{enumerate} 
    \item The inverse accuracy level $1 / \eps $ is an integer.
    
    \item An estimate ${\cal E} \in [(1 - \eps) \cdot \expar{ {\cal M}_r(S^*) }, \expar{ {\cal M}_r(S^*) }]$ for the optimal expected value ${\cal M}_r(S^*)$ is known in advance.
    
    \item The variables $X_1, \ldots, X_n$ are defined over the same support, ${\cal V} = \big\{ 0, \frac{\eps}{r} {\cal E}, \frac{2\eps}{r} {\cal E}, \ldots, \frac{\cal E}{r \eps } \big\}$.
\end{enumerate}
Given these assumptions, we guess the probability $\pr{X_i \leq v}$ for every unknown variable $X_i \in S^*$ and every $v \in \calV$. Formally, we guess estimates $\tilde{p}^{(i)}_{\leq v} \in \{2 \delta^2, 3 \delta^2, \ldots, 1 - \delta^2 \}$ for some $\delta= \delta(\eps)$ that satisfy $\prpar{ X_i \leq v } \leq \tilde{p}^{(i)}_{\leq v} \leq \prpar{ X_i \leq v } + 2\delta^2$. These estimates can be efficiently enumerated over, since $k$, $|\calV|$, and $\delta$ only depend on $\eps$. By exploiting a reduction to \Santa analogous to the one in Section~\ref{subsec:probemax_santa}, we compute a subset  $\tilde{S} \subseteq [n]$ of $k$ random variables, say $Y_1, \ldots, Y_k$, such that each $Y_i$ satisfies $\prpar{ Y_i \leq v } \leq (1 + \delta^2) \cdot \tilde{p}^{(i)}_{\leq v} + \delta^2$, for every $v \in \calV$. Finally, we establish the performance guarantee of $\tilde{S}$, showing that the latter inequalities are sufficient to argue that $\expar{ {\cal M}_r(\tilde{S}) } \geq (1 - \eps) \cdot \expar{ {\cal M}_r(S^*) }$, for a suitable choice of $\delta= \delta(\eps)$; the arguments are very similar to our analysis in Section~\ref{subsec:probemax_sufficient}.
\halmos
\endproof

\begin{lemma} \label{lem:TopRNonAdapProbe}
When $r > 1/\eps^3$, there exists an EPTAS for non-adaptive  \ProbeTopr.
\end{lemma} 
\proof{Proof  overview.}
For simplicity, we assume that the random variables $X_1, \ldots, X_n$ have a polynomially-sized support, $\calV = \{v_1, v_2, \ldots, v_{|\calV|} \}$, with the convention that $p_{ij} = \pr{X_i = v_j}$. The following standard LP-relaxation~\cite{GN-IPCO13} provides an upper bound on the optimal value $\expar{ {\cal M}_r(S^*) }$: 
\begin{equation}
\begin{array}{lll} \label{LP:ToprProbeMax}
\max \quad & {\displaystyle \sum_{i \in [n]} \sum_{j \in [|\calV|]} y_{ij} v_{j}}\\
\text{s.t.} \quad & {\displaystyle \sum_{i \in [n]} x_{i} \leq k} \\
& y_{ij}  \leq p_{ij} x_{i}	\qquad &\forall \, i \in [n],~ j \in [|\calV|]  \\
& {\displaystyle \sum_{i \in [n]} \sum_{j \in [|\calV|]} y_{ij} \leq r}\\
& y_{ij} \in [0,1] , \, x_i \in [0,1] \qquad \mbox{} & \forall \, i \in [n],~ j \in [|\calV|]  
\end{array}
\end{equation}
To verify that the optimal LP value is at least $\expar{ {\cal M}_r(S^*) }$, one can simply notice that a feasible solution is obtained by setting $x_i=1$ if and only if $i \in S^*$. In addition, $y_{ij} = \prpar{ i \in \mathrm{Top}_r( S^* ) , X_i = v_j }$ for every $j \in [|\calV|]$, if and only if $i \in S^*$. 

Now, let $(x^*,y^*)$ be an optimal fractional solution to \eqref{LP:ToprProbeMax}, and consider a random set $\tilde{S} \subseteq [n]$ that independently contains every random variable $X_i$ with probability $(1-\eps) \cdot x_i^*$. We first observe that $\tilde{S}$ consists of at most $k$ variables with high probability, due to having $k \geq r > 1 / \eps^3$. Next, we claim that $\expar{ {\cal M}_r(\tilde{S}) } = (1-O(\eps))\cdot \sum_{i \in [n]} \sum_{j \in [|\calV|]} y^*_{ij} v_{j}$. For this purpose, we construct a random set $R \subseteq \tilde{S}$ of size at most $r$, which implies $\expar{ {\cal M}_r(\tilde{S}) } \geq \expar{\sum_{i \in R} X_i}$, arguing that $\expar{\sum_{i \in R} X_i} =  (1-O(\eps))\cdot \sum_{i \in [n]} \sum_{j \in [|\calV|]} y^*_{ij} v_{j}$.

In particular, consider a random set $R$, created by independently picking every $X_i \in \tilde{S}$ with probability $(1-\eps) \cdot \frac{y_{ij}}{x_i p_{ij}} $, where $j$ is the index for which $X_i$ was observed to realize as $v_j$. Note that $\frac{y_{ij}}{x_i p_{ij}} \leq 1$, due to the second constraint of our LP relaxation. The expected size of $R$, over the randomness in $\tilde{S}$ and in picking each $X_i$, is
\begin{eqnarray*} 
(1-\eps) \cdot \sum_{i \in [n]}  \pr{X_i \in \tilde{S}} \cdot \sum_{j \in [|\calV|]} \pr{X_i = v_j} \cdot\frac{y_{ij}}{x_i p_{ij}}   & = & (1-\eps)^2 \cdot \sum_{i \in [n]} x_i \cdot \sum_{j \in [|\calV|]} p_{ij} \cdot \frac{y_{ij}}{x_i  p_{ij}}   \\
 &=&  (1-\eps)^2 \cdot \sum_{i \in [n]} \sum_{j \in [|\calV|]} y_{ij} \\
 & \leq & (1-\eps)^2 r,
\end{eqnarray*}
where the last inequality follows from the third LP constraint. Again by Chernoff bounds, the set of picked variables $R$ will be of size at most $r$ with high probability, since $r > 1/\eps^3$. It is also easy to verify that $\expar{\sum_{i \in R} X_i} = (1-O(\eps))\cdot \sum_{i \in [n]} \sum_{j \in [|\calV|]} y^*_{ij} v_{j}$.
\halmos
\endproof

\section{Future Directions.}

n this work, we design EPTASes for a number of fundamental stochastic combinatorial optimization problems. An immediate open question for future research is whether one or more of these problems admit  \emph{fully} polynomial time approximation schemes (FPTAS). That is, for any constant $\eps>0$, can we obtain a $(1+\eps)$-approximation in $\poly(n,1/\eps)$ time? We note that an FPTAS is impossible for our \Santa problem, since even with 3 machines, it already captures the strongly NP-hard 3-Partition problem \cite{GJ-Book}. Hence, we cannot hope to convert 
Theorem~\ref{thm:multidimInformal} into an FPTAS and a different approach would be needed to obtain an FPTAS (if possible) for the stochastic problems studied in this paper.

Another interesting direction is to obtain an EPTAS/FPTAS for multiple-element selection variants of the problems studied in this paper. 
In Section~\ref{sec:mmultItems}, we discuss how our framework can be leveraged when the goal is to maximize the sum of top-$r$ elements,  i.e., a uniform matroid of rank $r$. It is worth investigating whether similar results can be obtained for more general matroids (e.g., laminar matroids~\cite{ANSS-EC19}).



{\small
\bibliographystyle{alpha}
\bibliography{mor-bib}

\newcommand{\etalchar}[1]{$^{#1}$}
\begin{thebibliography}{MNPR20}

\bibitem[Ada11]{Adamczyk-IPL11}
Marek Adamczyk.
\newblock Improved analysis of the greedy algorithm for stochastic matching.
\newblock {\em Information Processing Letters}, 111(15):731--737, 2011.

\bibitem[AGM15]{AGM-ESA15}
Marek Adamczyk, Fabrizio Grandoni, and Joydeep Mukherjee.
\newblock Improved approximation algorithms for stochastic matching.
\newblock In {\em Proceedings of European Symposium on Algorithms}, pages
  1--12, 2015.

\bibitem[AKS17]{AKS-TALG17}
Chidambaram Annamalai, Christos Kalaitzis, and Ola Svensson.
\newblock Combinatorial algorithm for restricted max-min fair allocation.
\newblock {\em ACM Transactions on Algorithms}, 13(3):1--28, 2017.

\bibitem[Ala14]{Alaei-SICOMP14}
Saeed Alaei.
\newblock Bayesian combinatorial auctions: Expanding single buyer mechanisms to
  many buyers.
\newblock {\em SIAM Journal on Computing}, 43(2):930--972, 2014.

\bibitem[AN16]{AN16}
Arash Asadpour and Hamid Nazerzadeh.
\newblock Maximizing stochastic monotone submodular functions.
\newblock {\em Management Science}, 62(8):2374--2391, 2016.

\bibitem[ANS08]{ANS-WINE08}
Arash Asadpour, Hamid Nazerzadeh, and Amin Saberi.
\newblock Stochastic submodular maximization.
\newblock In {\em Proceedings of International Workshop on Internet and Network
  Economics}, pages 477--489, 2008.
\newblock Full version appears as~\cite{AN16}.

\bibitem[ANSS19]{ANSS-EC19}
Nima Anari, Rad Niazadeh, Amin Saberi, and Ali Shameli.
\newblock Nearly optimal pricing algorithms for production constrained and
  laminar bayesian selection.
\newblock In {\em Proceedings of the {ACM} Conference on Economics and
  Computation}, pages 91--92, 2019.

\bibitem[ASW16]{ASW14}
Marek Adamczyk, Maxim Sviridenko, and Justin Ward.
\newblock Submodular stochastic probing on matroids.
\newblock {\em Mathematics of Operations Research}, 41(3):1022--1038, 2016.

\bibitem[ASZ20]{ASZ-EC20}
Shipra Agrawal, Jay Sethuraman, and Xingyu Zhang.
\newblock On optimal ordering in the optimal stopping problem.
\newblock In {\em Proceedings of the {ACM} Conference on Economics and
  Computation}, pages 187--188, 2020.

\bibitem[BBB{\etalchar{+}}21]{Bonamy+21}
Marthe Bonamy, {\'E}douard Bonnet, Nicolas Bousquet, Pierre Charbit, Panos
  Giannopoulos, Eun~Jung Kim, Pawe{\l} Rz{{a}}{\.z}ewski, Florian Sikora, and
  St{\'e}phan Thomass{\'e}.
\newblock {EPTAS} and subexponential algorithm for maximum clique on disk and
  unit ball graphs.
\newblock {\em Journal of the ACM}, 68(2):1--38, 2021.

\bibitem[BCN{\etalchar{+}}18]{BCNSX-APPROX15}
Alok Baveja, Amit Chavan, Andrei Nikiforov, Aravind Srinivasan, and Pan Xu.
\newblock Improved bounds in stochastic matching and optimization.
\newblock {\em Algorithmica}, 80(11):3225--3252, 2018.

\bibitem[BD05]{BD-05}
Ivona Bez{\'a}kov{\'a} and Varsha Dani.
\newblock Allocating indivisible goods.
\newblock {\em ACM SIGecom Exchanges}, 5(3):11--18, 2005.

\bibitem[BFLL20]{BFLL-EC20}
Shant Boodaghians, Federico Fusco, Philip Lazos, and Stefano Leonardi.
\newblock Pandora's box problem with order constraints.
\newblock In {\em Proceedings of the {ACM} Conference on Economics and
  Computation}, pages 439--458, 2020.

\bibitem[BGK11]{BGK-SODA11}
Anand Bhalgat, Ashish Goel, and Sanjeev Khanna.
\newblock Improved approximation results for stochastic knapsack problems.
\newblock In {\em Proceedings of the Annual {ACM-SIAM} Symposium on Discrete
  Algorithms}, pages 1647--1665, 2011.

\bibitem[BGL{\etalchar{+}}12]{BGLMNR-Algorithmica12}
Nikhil Bansal, Anupam Gupta, Jian Li, Juli{\'a}n Mestre, Viswanath Nagarajan,
  and Atri Rudra.
\newblock When {LP} is the cure for your matching woes: Improved bounds for
  stochastic matchings.
\newblock {\em Algorithmica}, 63(4):733--762, 2012.

\bibitem[BIKK08]{BIKK-SIGecom08}
Moshe Babaioff, Nicole Immorlica, David Kempe, and Robert Kleinberg.
\newblock Online auctions and generalized secretary problems.
\newblock {\em SIGecom Exchanges}, 7(2), 2008.

\bibitem[BK19]{BK-EC19}
Hedyeh Beyhaghi and Robert Kleinberg.
\newblock Pandora's problem with nonobligatory inspection.
\newblock In {\em Proceedings of the {ACM} Conference on Economics and
  Computation}, pages 131--132, 2019.

\bibitem[BN15]{BN-IPCO14}
Nikhil Bansal and Viswanath Nagarajan.
\newblock On the adaptivity gap of stochastic orienteering.
\newblock {\em Mathematical Programming}, 154(1-2):145--172, 2015.

\bibitem[BS06]{BS-STOC06}
Nikhil Bansal and Maxim Sviridenko.
\newblock The santa claus problem.
\newblock In {\em Proceedings of the Annual ACM SIGACT Symposium on Theory of
  Computing}, pages 31--40, 2006.

\bibitem[BSZ19]{BSZ-Approx19}
Domagoj Bradac, Sahil Singla, and Goran Zuzic.
\newblock {(Near)} optimal adaptivity gaps for stochastic multi-value probing.
\newblock In {\em Proceedings of International Conference on Randomization and
  Computation}, pages 49:1--49:21, 2019.

\bibitem[CCK09]{CCK-FOCS09}
Deeparnab Chakrabarty, Julia Chuzhoy, and Sanjeev Khanna.
\newblock On allocating goods to maximize fairness.
\newblock In {\em Proceedings of the Annual IEEE Symposium on Foundations of
  Computer Science}, pages 107--116, 2009.

\bibitem[CGT{\etalchar{+}}20]{CGTTZ-FOCS20}
Shuchi Chawla, Evangelia Gergatsouli, Yifeng Teng, Christos Tzamos, and Ruimin
  Zhang.
\newblock Pandora's box with correlations: Learning and approximation.
\newblock In {\em Proceedings of the Annual IEEE Symposium on Foundations of
  Computer Science}, pages 1214--1225, 2020.

\bibitem[Che52]{Chernoff-52}
Herman Chernoff.
\newblock A measure of asymptotic efficiency for tests of a hypothesis based on
  the sum of observations.
\newblock {\em The Annals of Mathematical Statistics}, 23(4):493--507, 1952.

\bibitem[CHL{\etalchar{+}}16]{CHLLLL-NIPS16}
Wei Chen, Wei Hu, Fu~Li, Jian Li, Yu~Liu, and Pinyan Lu.
\newblock Combinatorial multi-armed bandit with general reward functions.
\newblock In {\em Proceedings of Annual Conference on Neural Information
  Processing Systems}, pages 1651--1659, 2016.

\bibitem[CHMS10]{CHMS-STOC10}
Shuchi Chawla, Jason~D. Hartline, David~L. Malec, and Balasubramanian Sivan.
\newblock Multi-parameter mechanism design and sequential posted pricing.
\newblock In {\em Proceedings of the Annual ACM SIGACT Symposium on Theory of
  Computing}, pages 311--320, 2010.

\bibitem[CIK{\etalchar{+}}09]{CIKMR-ICALP09}
Ning Chen, Nicole Immorlica, Anna~R. Karlin, Mohammad Mahdian, and Atri Rudra.
\newblock Approximating matches made in heaven.
\newblock In {\em Proceedings of the International Colloquium on Automata,
  Languages, and Programming}, pages 266--278, 2009.

\bibitem[DGV05]{DGV-SODA05}
Brian~C. Dean, Michel~X. Goemans, and Jan Vondr{\'{a}}k.
\newblock Adaptivity and approximation for stochastic packing problems.
\newblock In {\em Proceedings of the Annual {ACM-SIAM} Symposium on Discrete
  Algorithms}, pages 395--404, 2005.

\bibitem[DGV08]{DGV-FOCS04}
Brian~C. Dean, Michel~X. Goemans, and Jan Vondr{\'{a}}k.
\newblock Approximating the stochastic knapsack problem: The benefit of
  adaptivity.
\newblock {\em Mathematics of Operations Research}, 33(4):945--964, 2008.

\bibitem[Dyn63]{Dynkin-Journal63}
Eugene~B Dynkin.
\newblock The optimum choice of the instant for stopping a markov process.
\newblock {\em Soviet Mathematics}, 4:627--629, 1963.

\bibitem[EFGT20]{EFGT-EC20}
Tomer Ezra, Michal Feldman, Nick Gravin, and Zhihao~Gavin Tang.
\newblock Online stochastic max-weight matching: prophet inequality for vertex
  and edge arrival models.
\newblock In {\em Proceedings of the {ACM} Conference on Economics and
  Computation}, pages 769--787, 2020.

\bibitem[EHKS24]{EHKS-2024}
Soheil Ehsani, Mohammadtaghi Hajiaghayi, Thomas Kesselheim, and Sahil Singla.
\newblock Prophet secretary for combinatorial auctions and matroids.
\newblock {\em SIAM Journal on Computing}, 53(6):1641--1662, 2024.

\bibitem[Fei08]{Feige-SODA08}
Uriel Feige.
\newblock On allocations that maximize fairness.
\newblock In {\em Proceedings of the Annual {ACM-SIAM} Symposium on Discrete
  Algorithms}, pages 287--293, 2008.

\bibitem[FLRS11]{FominLRS11}
Fedor~V. Fomin, Daniel Lokshtanov, Venkatesh Raman, and Saket Saurabh.
\newblock Bidimensionality and {EPTAS}.
\newblock In {\em Proceedings of the Twenty-Second Annual {ACM-SIAM} Symposium
  on Discrete Algorithms}, pages 748--759, 2011.

\bibitem[FLX18]{FLX-ICALP18}
Hao Fu, Jian Li, and Pan Xu.
\newblock A {PTAS} for a class of stochastic dynamic programs.
\newblock In {\em Proceedings of International Colloquium on Automata,
  Languages, and Programming}, pages 56:1--56:14, 2018.

\bibitem[FSZ21]{FSZ-SODA16}
Moran Feldman, Ola Svensson, and Rico Zenklusen.
\newblock Online contention resolution schemes with applications to bayesian
  selection problems.
\newblock {\em {SIAM} Journal on Computing}, 50(2):255--300, 2021.

\bibitem[GGM10]{GGM-TALG10}
Ashish Goel, Sudipto Guha, and Kamesh Munagala.
\newblock How to probe for an extreme value.
\newblock {\em ACM Transactions on Algorithms}, 7(1):12, 2010.

\bibitem[GJ02]{GJ-Book}
Michael~R Garey and David~S Johnson.
\newblock {\em Computers and Intractability}, volume~29.
\newblock W.H. Freeman, New York, 2002.

\bibitem[GJSS19]{GJSS-IPCO19}
Anupam Gupta, Haotian Jiang, Ziv Scully, and Sahil Singla.
\newblock The markovian price of information.
\newblock In {\em Proceedings of the International Conference on Integer
  Programming and Combinatorial Optimization}, pages 233--246, 2019.

\bibitem[GKNR15]{GKNR-SODA12}
Anupam Gupta, Ravishankar Krishnaswamy, Viswanath Nagarajan, and R.~Ravi.
\newblock Running errands in time: Approximation algorithms for stochastic
  orienteering.
\newblock {\em Mathematics of Operations Research}, 40(1):56--79, 2015.

\bibitem[GKPS06]{GandhiKPS06}
Rajiv Gandhi, Samir Khuller, Srinivasan Parthasarathy, and Aravind Srinivasan.
\newblock Dependent rounding and its applications to approximation algorithms.
\newblock {\em Journal of the ACM}, 53(3):324--360, 2006.

\bibitem[GKS19]{GKS-SODA19}
Buddhima Gamlath, Sagar Kale, and Ola Svensson.
\newblock Beating greedy for stochastic bipartite matching.
\newblock In {\em Proceedings of the Annual {ACM-SIAM} Symposium on Discrete
  Algorithms}, pages 2841--2854, 2019.

\bibitem[GM09]{GM-ICALP09}
Sudipto Guha and Kamesh Munagala.
\newblock Multi-armed bandits with metric switching costs.
\newblock In {\em Proceedings of International Colloquium on Automata,
  Languages, and Programming}, pages 496--507, 2009.

\bibitem[GN13]{GN-IPCO13}
Anupam Gupta and Viswanath Nagarajan.
\newblock A stochastic probing problem with applications.
\newblock In {\em Proceedings of the International Conference on Integer
  Programming and Combinatorial Optimization}, pages 205--216, 2013.

\bibitem[GNS17]{GNS-SODA17}
Anupam Gupta, Viswanath Nagarajan, and Sahil Singla.
\newblock {Adaptivity Gaps for Stochastic Probing: Submodular and XOS
  Functions}.
\newblock In {\em Proceedings of the Annual {ACM-SIAM} Symposium on Discrete
  Algorithms}, pages 1688--1702, 2017.

\bibitem[GS20]{GS-arXiv20}
Anupam Gupta and Sahil Singla.
\newblock Random-order models.
\newblock In Tim Roughgarden, editor, {\em Beyond the Worst-Case Analysis of
  Algorithms}. Cambridge University Press, 2020.

\bibitem[Hil83]{Hill-Journal83}
TP~Hill.
\newblock Prophet inequalities and order selection in optimal stopping
  problems.
\newblock {\em Proceedings of the American Mathematical Society},
  88(1):131--137, 1983.

\bibitem[HKS07]{HKS-AAAI07}
Mohammad~Taghi Hajiaghayi, Robert~D. Kleinberg, and Tuomas Sandholm.
\newblock Automated online mechanism design and prophet inequalities.
\newblock In {\em Proceedings of the Twenty-Second AAAI Conference on
  Artificial Intelligence}, pages 58--65, 2007.

\bibitem[HL04]{HassinL04}
Refael Hassin and Asaf Levin.
\newblock An efficient polynomial time approximation scheme for the constrained
  minimum spanning tree problem using matroid intersection.
\newblock {\em {SIAM} Journal on Computing}, 33(2):261--268, 2004.

\bibitem[Jan10]{Jansen10}
Klaus Jansen.
\newblock An {EPTAS} for scheduling jobs on uniform processors: Using an {MILP}
  relaxation with a constant number of integral variables.
\newblock {\em {SIAM} Journal on Discrete Mathematics}, 24(2):457--485, 2010.

\bibitem[JLLS20]{JLLS-ITCS20}
Haotian Jiang, Jian Li, Daogao Liu, and Sahil Singla.
\newblock Algorithms and adaptivity gaps for stochastic $k$-{TSP}.
\newblock In {\em Proceedings of 11th Innovations in Theoretical Computer
  Science Conference}, pages 45:1--45:25, 2020.

\bibitem[KS77]{Krengel-Journal77}
Ulrich Krengel and Louis Sucheston.
\newblock Semiamarts and finite values.
\newblock {\em Bulletin of the American Mathematical Society}, 83(4):745--747,
  1977.

\bibitem[KS78]{Krengel-Journal78}
Ulrich Krengel and Louis Sucheston.
\newblock On semiamarts, amarts, and processes with finite value.
\newblock {\em Advances in Probability}, 4:197--266, 1978.

\bibitem[KW19]{KW-STOC12}
Robert Kleinberg and S.~Matthew Weinberg.
\newblock Matroid prophet inequalities and applications to multi-dimensional
  mechanism design.
\newblock {\em Games and Economic Behavior,}, 113:97--115, 2019.

\bibitem[KWW16]{KWW-EC16}
Robert~D. Kleinberg, Bo~Waggoner, and E.~Glen Weyl.
\newblock Descending price optimally coordinates search.
\newblock In {\em Proceedings of the {ACM} Conference on Economics and
  Computation}, pages 23--24, 2016.

\bibitem[LLP{\etalchar{+}}21]{LLPSS-arXiv20}
Allen Liu, Renato~Paes Leme, Martin P{\'{a}}l, Jon Schneider, and
  Balasubramanian Sivan.
\newblock Variable decomposition for prophet inequalities and optimal ordering.
\newblock In {\em Proceedings of the 22nd {ACM} Conference on Economics and
  Computation}, page 692, 2021.

\bibitem[Luc17]{Lucier-SIGecom17}
Brendan Lucier.
\newblock An economic view of prophet inequalities.
\newblock {\em SIGecom Exchanges}, 16(1):24--47, 2017.

\bibitem[LY13]{LiYuan-STOC13}
Jian Li and Wen Yuan.
\newblock Stochastic combinatorial optimization via poisson approximation.
\newblock In {\em Proceedings of the Annual ACM SIGACT Symposium on Theory of
  Computing}, pages 971--980, 2013.

\bibitem[Ma18]{Ma-SODA14}
Will Ma.
\newblock Improvements and generalizations of stochastic knapsack and markovian
  bandits approximation algorithms.
\newblock {\em Mathematics of Operations Research}, 43(3):789--812, 2018.

\bibitem[MNPR20]{MNPR-NeurIPS20}
Aranyak Mehta, Uri Nadav, Alexandros Psomas, and Aviad Rubinstein.
\newblock Hitting the high notes: Subset selection for maximizing expected
  order statistics.
\newblock {\em Advances in Neural Information Processing Systems},
  33:15800--15810, 2020.

\bibitem[NWF78]{NWF-MP78}
George~L Nemhauser, Laurence~A Wolsey, and Marshall~L Fisher.
\newblock {An analysis of approximations for maximizing submodular set
  functions-I}.
\newblock {\em Mathematical programming}, 14(1):265--294, 1978.

\bibitem[RS17]{RS-SODA17}
Aviad Rubinstein and Sahil Singla.
\newblock Combinatorial prophet inequalities.
\newblock In {\em Proceedings of the Annual {ACM-SIAM} Symposium on Discrete
  Algorithms}, pages 1671--1687, 2017.

\bibitem[Rub16]{Rubinstein-STOC16}
Aviad Rubinstein.
\newblock Beyond matroids: Secretary problem and prophet inequality with
  general constraints.
\newblock In {\em Proceedings of the Annual ACM SIGACT Symposium on Theory of
  Computing}, pages 324--332, 2016.

\bibitem[SC84]{Samuel-Annals84}
Ester Samuel-Cahn.
\newblock Comparison of threshold stop rules and maximum for independent
  nonnegative random variables.
\newblock {\em The Annals of Probability}, 12(4):1213--1216, 1984.

\bibitem[Sin18a]{Singla-Thesis18}
Sahil Singla.
\newblock {\em Combinatorial Optimization Under Uncertainty: Probing and
  Stopping-Time Algorithms}.
\newblock PhD thesis, Carnegie Mellon University, 2018.

\bibitem[Sin18b]{Singla-SODA18}
Sahil Singla.
\newblock The price of information in combinatorial optimization.
\newblock In {\em Proceedings of the Annual {ACM-SIAM} Symposium on Discrete
  Algorithms}, pages 2523--2532, 2018.

\bibitem[SS20]{SS-arXiv20}
Danny Segev and Sahil Singla.
\newblock Efficient approximation schemes for stochastic probing and prophet
  problems.
\newblock {\em CoRR}, abs/2007.13121, 2020.

\bibitem[Wei79]{Weitzman-Econ79}
Martin~L. Weitzman.
\newblock Optimal search for the best alternative.
\newblock {\em Econometrica: Journal of the Econometric Society},
  47(3):641--654, 1979.

\bibitem[Yan11]{Yan-SODA11}
Qiqi Yan.
\newblock Mechanism design via correlation gap.
\newblock In {\em Proceedings of the Annual {ACM-SIAM} Symposium on Discrete
  Algorithms}, pages 710--719, 2011.

\end{thebibliography}
}


\end{document}